\documentclass[runningheads, orivec]{llncs}

\usepackage[T1]{fontenc}

\usepackage{amsfonts, amssymb}

\usepackage{thm-restate} %

\usepackage{booktabs}

\usepackage{mathtools}

\usepackage{enumitem}

\usepackage{blindtext}

\usepackage{graphicx}

\usepackage{listings}

\usepackage{multirow}

\usepackage{stmaryrd}

\usepackage{xcolor}

\usepackage{tikz}

\usepackage{hyperref}

\usepackage[capitalize, nameinlink]{cleveref}
\crefname{section}{Sect.}{Sects.}
\crefname{lemma}{Lem.}{Lemms.}
\crefname{theorem}{Thm.}{Thms.}

\renewcommand{\orcidID}[1]{\href{#1}{\includegraphics[height=9pt]{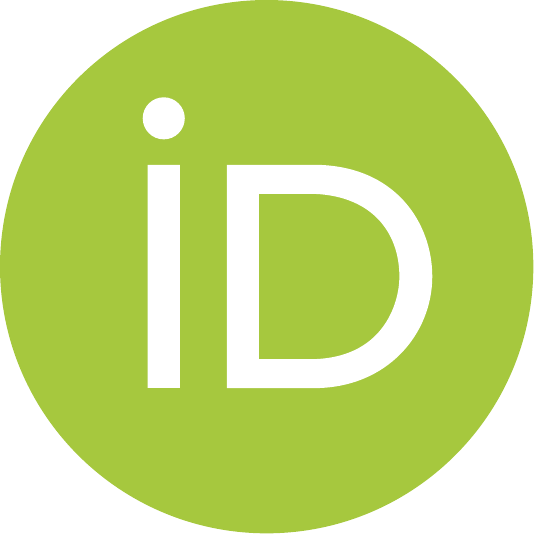}}}

\definecolor{cispa-blue}{HTML}{009CC4}
\definecolor{cispa-blue-80}{HTML}{33B0D0}
\definecolor{cispa-blue-60}{HTML}{66C4DC}
\definecolor{cispa-blue-40}{HTML}{99D7E7}
\definecolor{cispa-blue-20}{HTML}{D5ECF6}
\definecolor{cispa-dark-blue}{HTML}{002864}
\definecolor{cispa-medium-blue}{HTML}{018DCA}
\definecolor{cispa-bright-blue}{HTML}{00C8FF}
\definecolor{cispa-anthracite}{HTML}{1C1C1C}
\definecolor{cispa-dark-gray}{HTML}{2B2B2B}
\definecolor{cispa-light-gray}{HTML}{C9C9C9}
\definecolor{cispa-grid-gray-light}{HTML}{E8E8E8}
\definecolor{cispa-offwhite}{HTML}{F5F5F5}
\definecolor{cispa-green}{HTML}{008040}
\definecolor{cispa-red}{HTML}{EB3300}
\definecolor{cispa-orange}{HTML}{FF7F32}
\definecolor{cispa-yellow}{HTML}{FFC600}

\usetikzlibrary{
	positioning,
	arrows,
	automata,
	patterns,
}

\tikzstyle{lts}=[auto, ->, >=stealth', font=\small]
\tikzstyle{every picture}=[style=lts]

\tikzstyle{every node}=[initial text=]

\tikzstyle{location}=[
rectangle,
rounded corners,
minimum size=12pt,
draw=cispa-anthracite,
fill=cispa-blue-20,
text=cispa-anthracite,
inner sep=3.5pt
]
\tikzstyle{locguard}=[
fill=cispa-bright-blue,
inner sep=0.5pt,
node distance=2pt,
yshift=0.5pt,
font=\footnotesize,
]
\tikzstyle{final}=[double]
\tikzstyle{init}=[initial, initial distance=0.35cm]

\hypersetup{
	colorlinks = true,
	breaklinks = true,
	citecolor   = cispa-dark-blue,
	linkcolor   = cispa-dark-blue,
	urlcolor    = cispa-medium-blue,
}

\newcommand{\csystem}{\mathcal{C}} %
\newcommand{\csystemdef}{(C, Q, \mathcal{T})} %
\newcommand{\csystemabstract}{\csystem_\alpha} %
\newcommand{\Dmin}{\Delta_{\text{min}}}  %

\newcommand{\N}{\mathbb{N}}                    %
\newcommand{\set}[1]{\left\{#1\right\}}        %
\newcommand{\trans}[1]{\xrightarrow{#1}}       %
\newcommand{\supp}[1]{{\llbracket#1\rrbracket}}  %

\newcommand{\PSPACE}{\textsf{PSPACE}\xspace}

\newcommand{\NP}{\textsf{NP}\xspace}
\newcommand{\Polytime}{\textsf{PTIME}\xspace}
\newcommand{\NPSPACE}{\textsf{NPSPACE}\xspace}

\newcommand{\wqo}{\preceq_0}                    %

\newcommand{\control}{\ensuremath{\mathsf{ctrl}}}

\newcommand{\card}[1]{{\left| {#1} \right|}}

\newcommand{\traces}[1]{\ensuremath{\mathsf{Traces}(#1)}}
\newcommand{\tracesinf}[1]{\ensuremath{\mathsf{Traces}_\infty(#1)}}

\newcommand{\conf}[1]{\ensuremath{\mathbf{#1}}}
\newcommand{\absconf}[1]{\ensuremath{\mathbf{#1^\alpha}}}
\newcommand{\config}{\ensuremath{(c,\conf{v})}\xspace}
\newcommand{\configPrime}{\ensuremath{(c',\conf{v}')}\xspace}
\newcommand{\loc}{\ensuremath{q}}
\newcommand{\Loc}{\ensuremath{Q}}

\newcommand{\initloc}{\ensuremath{\hat{\loc}}}
\newcommand{\pre}[1]{\mathit{pre}(#1)}
\newcommand{\post}[1]{\mathit{post}(#1)}

\newcommand{\smartpar}[1]{\medskip \noindent \textbf{#1}.}
\newcommand{\existsGuard}{\ensuremath{G_\exists}}
\newcommand{\forallGuard}{\ensuremath{G_\forall}}
\newcommand{\allGuard}{\ensuremath{\forallGuard}}
\newcommand{\act}{\ensuremath{a}\xspace}

\newcommand{\disjunctive}[1]{disjunctive system}

\newcommand{\node}{\ensuremath{v}}

\newcommand{\Transitions}{\ensuremath{\mathcal{T}}}

\newcommand{\autom}{\ensuremath{M}}
\newcommand{\automStates}{\ensuremath{\Loc}}
\newcommand{\automAlphabet}{\ensuremath{\Sigma_{\autom}}}
\newcommand{\automTrans}{\ensuremath{\Transitions}}
\newcommand{\automInitState}{\ensuremath{q^{\mathsf{init}}}}
\newcommand{\automAcceptStates}{\ensuremath{Q^{\mathsf{final}}}}
\newcommand{\word}{\ensuremath{w}}

\begin{document}
\title{Parameterized Verification of Systems with Precise (0,1)-Counter Abstraction}
\titlerunning{Parameterized Systems with Precise (0,1)-Counter Abstraction}

\author{%
	Paul Eichler\inst{1}\orcidID{https://orcid.org/0009-0008-6117-318X} \and
	Swen Jacobs\inst{1}\orcidID{https://orcid.org/0000-0002-9051-4050} \and
	Chana {Weil-Kennedy}\inst{2}\orcidID{https://orcid.org/0000-0002-1351-8824}
}
\authorrunning{P. Eichler, S. Jacobs, C. {Weil-Kennedy}}

\institute{
	CISPA Helmholtz Center for Information Security, Saarbr\"ucken, Germany
	\email{\{\href{mailto:paul.eichler@cispa.de}{paul.eichler}, \href{mailto:jacobs@cispa.de}{jacobs}\}@cispa.de}\\
	IMDEA Software Institute, Madrid, Spain\\
	\email{\href{mailto:chana.weilkennedy@imdea.org}{chana.weilkennedy@imdea.org}}
}

\maketitle

\begin{abstract}
	We introduce a new framework for verifying systems with a parametric number of concurrently running processes. The systems we consider are well-structured with respect to a specific well-quasi order. This allows us to decide a wide range of verification problems, including control-state reachability, coverability, and target, in a fixed finite abstraction of the infinite state-space, called a 01-counter system. We show that several systems from the parameterized verification literature fall into this class, including reconfigurable broadcast networks (or systems with lossy broadcast), \disjunctive{}s, synchronizations and systems with a fixed number of shared finite-domain variables. Our framework provides a simple and unified explanation for the properties of these systems, which have so far been investigated separately. Additionally, it extends and improves on a range of the existing results, and gives rise to other systems with similar properties.

\end{abstract}

\section{Introduction}
\label{sec:introduction}

Concurrent systems often consist of an arbitrary number of uniform user processes running in parallel, 
possibly with a distinguished controller process. 
Given a description of the user and controller protocols and a desired property, the \emph{parameterized model checking problem} (PMCP) is to decide whether the property holds in the system, regardless of the number of user processes.
The PMCP is well-known to be undecidable in general~\cite{Apt86}, even when the property is control-state reachability and all processes are finite-state~\cite{Suzuki88}.
However, a long line of research has valiantly strived for the identification of decidable fragments that support interesting models and properties~\cite{GermanS92,DBLP:conf/popl/EmersonN95,abdulla1996general,EsparzaFM99,EmersonK00,Esparza14,AminofKRSV18,DBLP:series/synthesis/2015Bloem}.

One of the most prominent techniques for the identification of fragments with decidable PMCP are well-structured transition systems (WSTS)~\cite{DBLP:journals/iandc/Finkel90,abdulla1996general,AbdullaCJT00,FinkelS01}.
The WSTS framework puts a number of restrictions on the system, most importantly the compatibility of its transition relation with a well-quasi order (wqo) on its (infinite) set of states, which in turn allows to decide some PMCP problems, including coverability.
However, while many of the works on parameterized verification share certain techniques, systems with different communication primitives have usually been studied separately, and it is hard to keep an overview of which problems are decidable for which class of systems, and why.

In this paper, we show that a range of systems, previously studied using different techniques, can be unified in a single framework.
Our framework gives a surprisingly simple explanation of existing decidability results for these systems, extends both the class of systems and the types of properties that can be verified, and allows us to prove previously unknown complexity bounds for some of these problems.
While the main condition of our framework resembles that of WSTS, i.e., compatibility of transitions with a wqo, we do not make use of any WSTS techniques.
Instead, we show that  $(0,1)$-counter abstraction (or simply $01$-abstraction), i.e., a binary abstraction that does not count the number of processes in a given state, but only distinguishes whether it is occupied or not, is precise for all systems satisfying the condition.
This abstraction is not only fixed for the whole class of systems, but may also be much more concise than the abstraction obtained by using WSTS techniques.
The wqo $\wqo$ we consider is an extension of the ``standard'' wqo for component-based systems, in which two configurations of the system are only comparable if they agree on which local states currently are occupied (by at least one process), and which are not occupied by any process. 

\smartpar{Parameterized Systems and Related Work}
The systems we consider are based on one control process and an arbitrary number of identical user processes.
Processes change state synchronously according to a step relation,
usually based on local transitions that may be synchronized based on transition labels.
In particular, our framework supports the following communication (or synchronization) primitives from the parameterized verification literature:
\begin{itemize}[topsep=0pt,parsep=1pt]
    \item  \emph{Lossy broadcast}~\cite{DelzannoSZ12}, where processes can send broadcast messages that may or may not be received by the other processes. 
		This model is equivalent to the widely studied system model of \emph{reconfigurable broadcast networks} (RBN)~\cite{DelzannoSTZ12,BertrandFS14,Balasubramanian18,Balasubramanian22}, where processes communicate via broadcast to their neighbors in the underlying communication topology, which can reconfigure at any time. 
    Here, we frame them as lossy broadcast in a clique topology, since we also assume all other systems to be arranged in a clique topology.
    
		\item \emph{Disjunctive guards}~\cite{EmersonK00}, where transitions of a process depend on the existence of another process that is in a certain local state. 
		Systems with disjunctive guards (or: \emph{disjunctive systems}) have been studied extensively in the literature~\cite{EmersonK00,EmersonK03,AusserlechnerJK16,JacobsS18,AminofKRSV18,JacobsSV22}. 
    We note that this model is equivalent to immediate observation (IO) protocols~\cite{EsparzaRW19}, a subclass of population protocols~\cite{AngluinAER07}, where a process observes the state of another process and changes its own state accordingly. 
    IO protocols are also known to be equivalent to the restriction of RBN in which broadcast transitions must be self-loops \cite{Gandalf21}. 
		
    \item \emph{Synchronization}, where transitions are labeled with actions and in every step of the system all processes synchronize on the same action.
		This model is studied for example in the context of controller synthesis~\cite{BertrandDGGG19,ColcombetFO21}. 
    There, the goal is 
    to decide if, for a given protocol followed by a parametric number of processes, a controller strategy exists that eventually puts all processes in the final state $f$.
    In \cite{ColcombetFO21}, the problem is posed in a stochastic setting.
    Synchronization protocols may be seen as a restriction of (non-lossy) broadcast protocols~\cite{EmersonN98,EsparzaFM99}.
    \item \emph{Asynchronous shared memory} (ASM)~\cite{EsparzaGM16} allows processes to communicate through finite-domain shared variables, but without locks and non-trivial read-modify-write operations, i.e., a transition cannot read and write a variable simultaneously.
		ASM systems (also called register protocols~\cite{BouyerMRSS16}) are known to be equivalent to RBN with regard to reachability properties~\cite{Gandalf21}.
	In \cite{EsparzaGM16} the authors go beyond what we consider in this work, as they consider that processes can also be pushdown machines or even Turing machines, and show that decidability can be preserved under certain restrictions.
\end{itemize}

\noindent Considering related verification techniques, 
close to ours in spirit is $(0,1,\infty)$-counter abstraction~\cite{PnueliXZ02}, with the crucial difference that their technique is approximative, while ours is precise for the systems we consider.
Additionally, $01$-counter abstraction has already been used for parameterized verification and repair in previous work~\cite{JacobsSV22,BaumeisterEJSV24}, but for more restricted classes of systems and, again, with correctness arguments specific to these classes.
In contrast, we provide a general criterion for correctness of the abstraction for a much broader class of systems and properties.

{In addition, there has been a lot of work on the parameterized verification of systems with more powerful communication primitives, such as pairwise rendezvous~\cite{GermanS92,AminofKRSV18} or (non-lossy) broadcast~\cite{EsparzaFM99}.
While these also fall into the class of WSTS, they are not compatible with the wqo $\wqo$, and $01$-abstraction is not precise for them.
Accordingly, the complexity of parameterized verification problems is in general much higher for these systems.
}

\smartpar{Contributions}
We introduce a common framework for the verification of parameterized systems that are well-structured with respect to the  wqo $\wqo$. 
\begin{itemize}
\item We prove that for all such systems, $01$-abstraction is sound and complete for safety properties, and that 
 lossy broadcast protocols, \disjunctive{}s, synchronization protocols, and ASM 
 fall into this class, as well as 
 systems based on a novel \emph{guarded} synchronization primitive, and 
 systems with combinations of these primitives (\cref{sec:basic}).

\item We show that a cardinality reachability (CRP) problem, which subsumes classical parameterized problems like coverability and target, is \PSPACE-complete for our class of systems
(\cref{sec:reach}).

\item We show how the $01$-abstraction can be leveraged to decide finite trace properties of a fixed number of processes in the parameterized system, and slightly improve known results on properties over infinite traces for \disjunctive{}s (\cref{sec:trace}).

\item We show that under modest additional assumptions on the systems, the complexity of the CRP is significantly lower (\cref{sec:small-tcs}).
\end{itemize}

\section{Preliminaries}
\label{sec:prelim}
\smartpar{Multisets}
A \emph{multiset} on a finite set \(E\) is a mapping \(C \colon E \rightarrow \N\), i.e. for any $e\in E$, \(C(e)\) denotes the number of occurrences of element \(e\) in \(C\).
We sometimes consider $C$ as a vector of length the cardinality of $E$, and denote it as $\conf{c} \in \N^E$.
Given $e \in E$, we denote by $\conf{e}$ the multiset consisting of one occurrence of element $e$. 
Operations on \(\N\) like addition or comparison are extended to multisets by defining them component-wise on each element of \(E\).
Subtraction is allowed in the following way: if $\conf{c},\conf{d}$ are multisets on set $E$ then for all $e\in E$, $(\conf{c}-\conf{d})(e)=\max (\conf{c}(e)-\conf{d}(e),0)$.
We call $|\conf{c}| = \sum_{e\in E} \conf{c}(e)$ the \emph{size} of $\conf{c}$.
The support $\supp{\conf{c}}$ of $\conf{c}$ is the set of elements $e\in E$ such that $\conf{c}(e)\ge1$.

\smartpar{Counter System}
Intuitively, a counter system explicitly keeps track of the state of the controller process, and for user processes keeps track of \emph{how many} user processes are in which local state.
Let us formalize this idea.
\begin{definition}
A \emph{counter system (CS)} is a triple $\csystem = \csystemdef$
where $C$ is the finite set of states of the controller,
$Q$ is the finite set of states of the users
 and $\mathcal{T}$ is the step relation 
such that $\mathcal{T} \subseteq (C \times \N^Q) \times (C \times \N^Q)$, where $|\conf{v}|=|\conf{v}'|$ whenever $((c,\conf{v}),(c',\conf{v}')) \in \mathcal{T}$, i.e., steps are size-preserving.
A \emph{configuration} of $\csystem$ is a pair $(c, \conf{v}) \in C \times \N^Q$.
We may fix initial states $c_0 \in C$ and $Q_0 \subseteq Q$; 
an \emph{initial configuration} is then any $(c_0,\conf{v}_0)$ such that $\conf{v}_0(q)=0$ for all $q$ not in $Q_0$.
The \emph{size} of a configuration is $|(c,\conf{v})| = 1+ |\conf{v}|$.
\end{definition}

If $\left( (c,\conf{v}) , (c',\conf{v}') \right) \in \mathcal{T}$ then we say there is a \emph{step} from $(c,\conf{v})$ to $(c',\conf{v}')$,
also denoted $(c,\conf{v}) \trans{} (c',\conf{v}')$.
We denote by $\trans{*}$ the reflexive and transitive closure of the step relation. 
A sequence of steps is called a \emph{path} of $\csystem$. 
A path is a \emph{run} if it starts in an initial configuration.
A configuration $(c,\conf{v})$ is \emph{reachable} if there is a run that ends in $(c,\conf{v})$.

\begin{remark}
In contrast to some of the results in this area, our model supports an additional distinguished controller process, which may execute a different protocol than the user processes.
It is known that in some settings the model with a controller is strictly more expressive than the model without~\cite{AminofKRSV18}.

Moreover, since our model also supports multiple initial states for the user processes, our results extend to the case of any fixed number of distinguished processes, and any fixed number of different types of user processes.%
\footnote{To see this, note first that if a system has multiple controllers, we can encode all of them as a single controller by simply considering their (finite-state) product.
To support $k$ different types of user processes with state sets $Q_1,\ldots,Q_k$ such that $Q_i \cap Q_j = \emptyset$ for all $i \neq j$, we simply construct one big user process with state set $Q_1 \cup \cdots \cup Q_k$, and similarly let the union of all individual initial states be the initial states of the constructed system.}
To keep notation simple, we will use a single controller and a single type of user process throughout the paper.
\end{remark}

\smartpar{Well-quasi Order}
Let $S$ be the (infinite) set of configurations of a CS.
A \emph{well-quasi order (wqo)} on $S$ is a relation ${\preceq} \subseteq S \times S$ that is reflexive and transitive, and is such that every infinite sequence  $s_0, s_1, \ldots$ of elements from $S$ contains an increasing pair $s_i \preceq s_j$ with $i < j$.

A wqo commonly used on configurations of a CS is defined as follows: 
\begin{center}
$(c,\conf{v}) \preceq (d, \conf{w}) \Leftrightarrow \left( c=d \land  \forall q \in Q{:} \conf{v}(q) \leq \conf{w}(q) \right)$
\end{center}
We define our wqo $\wqo$ as the following refinement of $\preceq$:
\begin{center}
$(c,\conf{v}) \wqo (d, \conf{w}) \Leftrightarrow$ $\left( (c,\conf{v}) \preceq (d, \conf{w}) \land \forall q \in Q{:} \left(\conf{v}(q)=0 \Leftrightarrow \conf{w}(q)=0 \right) \right)$
\end{center}

\smartpar{Compatibility}
We say that a CS $\csystem$ is \emph{forward $\preceq$-compatible}\footnote{This is sometimes called \emph{strong compatibility} in the literature.} for a wqo $\preceq$ if whenever there is a step $(c,\conf{v}) \rightarrow (c',\conf{v}')$ and $(c,\conf{v}) \preceq (d,\conf{w})$, then there exists a %
step $(d,\conf{w}) \rightarrow (d',\conf{w}')$ with $(c',\conf{v}') \preceq (d',\conf{w}')$.
We say $\csystem$ is \emph{backward $\preceq$-compatible} if whenever there is a step $(c,\conf{v}) \rightarrow (c',\conf{v}')$ and $(c',\conf{v}') \preceq (d',\conf{w}')$, then there exists a step $(d,\conf{w}) \rightarrow (d',\conf{w}')$ with $(c,\conf{v}) \preceq (d,\conf{w})$.
$\csystem$ is  \emph{fully $\preceq$-compatible} if it is forward and backward $\preceq$-compatible.

\smartpar{01-Counter System}
The idea of the $(0,1)$-counter abstraction is that we only distinguish whether a given local state is occupied or not.
This is formalized through an abstraction function
 $\alpha: C \times \mathbb{N}^Q \rightarrow C \times \{0,1\}^Q$  such that $\alpha(c,\conf{v}) = (c,\absconf{v})$, where $\absconf{v}(q)=1$ if $\conf{v}(q) \geq 1$ and $\absconf{v}(q)=0$ if $\conf{v}(q) =0$.
We define the abstraction of a given CS via $\alpha$.

\begin{definition}
The \emph{$01$-counter system} ($01$-CS)  of $\csystem = (C, Q, \mathcal{T})$
is the tuple $\csystemabstract = (C \times \{0,1\}^Q, \mathcal{T}_\alpha)$, where $\mathcal{T}_\alpha \subseteq (C \times \{0,1\}^Q) \times (C \times \{0,1\}^Q)$ is such that $(c,\absconf{v}) \trans{} (c',\absconf{v'}) \in \mathcal{T}_\alpha$ if there exists a concrete step $(c,\conf{v}) \rightarrow (c',\conf{v}') \in \mathcal{T}$ with $\alpha(c,\conf{v})=(c,\absconf{v})$ and $\alpha(c',\conf{v}')=(c',\absconf{v'})$.
Given initial states $c_0 \in C$ and $Q_0 \subseteq Q$ of $\csystem$, 
an initial configuration of $\csystemabstract$ is any $(c_0,\absconf{v})$ such that $\absconf{v}(q)=0$ for all $q$ not in $Q_0$.
\end{definition}

\begin{remark}
Unlike for CSs,
in a 01-CS it is not the case that steps occur only between configurations of the same size.
For example, we may have $(c,\conf{v}) \trans{} (c,\conf{v}')$ in a CS $\csystem$
 where $\conf{v}(q)=1$ for all states $q\in Q$, and the step sends all the user processes to a state $p$ in $\conf{v'}$.
 Then the corresponding 01-CS $\csystemabstract$ has a step $(c,\absconf{v}) \trans{} (c,\absconf{v'})$
 such that $\conf{v}=\absconf{v}$ and $\absconf{v'}(q)=1$ if $q=p$ and 0 elsewhere, i.e.,
 $|\absconf{v}|=|Q|+1$ while $|\absconf{v'}|=1$. 
\end{remark}

The following link between the wqo $\wqo$ and 01-CSs follows directly from the definition.

\begin{lemma}
\label{lm:wqo-01-counter-system}
Let $\csystem=(C,Q,\mathcal{T})$ a CS, 
and $(c,\absconf{v})$ a configuration in $C \times \{0,1\}^Q$.
Let $(d,\conf{w})$ a configuration in $C \times \N^Q$.
Then $(c,\absconf{v}) \wqo (d,\conf{w})$ if and only if $\alpha(d,\conf{w})=(c,\absconf{v})$.
\end{lemma}

\subsection{Types of Steps}
\label{sec:transition-types}
We formally define the communication primitives mentioned in Section \ref{sec:introduction}.
Steps between configurations are defined
as multisets of (local) transitions that are taken by different processes at the same time, i.e.,
every process takes at most one transition in a step.
We say a process in configuration $(c, \conf{v})$ \emph{takes a transition} $p \trans{} q$ 
if the process moves from state $p$ to $q$, resulting in a new configuration $(c', \conf{v}')$
equal to $(c,\conf{v}-\conf{p}+\conf{q})$ if $p,q \in Q$ and $\conf{v}(p)\ge 1$,
and equal to $(q,\conf{v})$ if $p,q \in C$ and $c=p$.
The conditions $\conf{v}(p)\ge 1$ and $c=p$ ensure that there is a process in $(c, \conf{v})$ that can take the transition.

The following definitions of transition types and the steps they induce 
come from the literature on systems in which all steps are of one type, 
e.g., RBN~\cite{DelzannoSTZ12,DelzannoSZ12} in which all steps are lossy broadcasts.
Note that we allow the same transition to be taken by more than one process in a single step, even if the classical definitions would consider this several consecutive steps.
We discuss the resulting differences after formally defining our steps.

\smartpar{Internal}
\emph{Internal transitions} are of the form 
$(p,q)$ with $p,q \in C$ or $p,q \in Q$,
also denoted $p \trans{} q$.
These induce \emph{internal steps} 
where one or more processes take the same transition $p \trans{} q$,
i.e., $(c,\conf{v}) \trans{} (c,\conf{v}-i \cdot \conf{p}+i \cdot \conf{q})$ if $p,q \in Q$ and $\conf{v}(p)\ge i$,
and $(p,\conf{v}) \trans{} (q,\conf{v})$ if $p,q \in C$.

\smartpar{Lossy broadcast~\cite{DelzannoSZ12}}
Let $\Sigma$ be a finite alphabet.
\emph{Lossy broadcast transitions} are of the form 
$(p,l,q)$ with $l \in \set{!a,?a \ | \ a \in \Sigma}$ 
and $p,q \in C$ or $p,q \in Q$.
We sometimes denote a transition $(p,l,q)$ by $p\trans{l} q$.
Transitions with $l$ of the form $!a$ are \emph{broadcast transitions},
and transitions with $l$ of the form $?a$ are \emph{receive transitions}.
A \emph{lossy broadcast step} 
from a configuration  $(c,\conf{v})$
is made up of one or more processes taking the same broadcast transition $p\trans{!a} p'$,
and an arbitrary number $k\ge 0$ of processes taking receive transitions $p_1 \trans{?a} p'_1, \ldots, p_k \trans{?a} p'_k$.
If $(c',\conf{v}')$ is the resulting configuration, we denote the step by $(c,\conf{v}) \trans{} (c',\conf{v}')$.

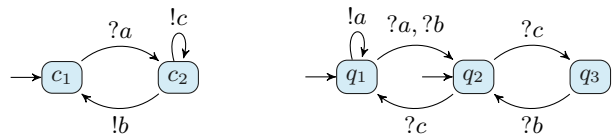
\begin{figure}[h]
  \centering
	\begin{tikzpicture}[font=\footnotesize]
      \node[location,initial] (a1) {$c_1$};
	  \node[location] (a2) [right =of a1] {$c_2$};

      \path[->]
      (a1) edge[bend left = 40] node[above] {$?a$} (a2)
      (a2) edge[bend left = 40] node[below] {$!b$} (a1)
      (a2) edge[loop above] node[above] {$!c$} (a2)
			;
    \end{tikzpicture}
	\hspace{3em}
    \begin{tikzpicture}[font=\footnotesize]
      \node[location,initial] (a1) {$q_1$};
	  \node[location,initial] (a2) [right =of a1] {$q_2$};
	  \node[location] (a3) [right =of a2] {$q_3$};

      \path[->]
      (a1) edge[bend left = 40] node[above] {$?a,? b$} (a2)
      (a2) edge[bend left = 40] node[below] {$?c$} (a1)
      (a3) edge[bend left = 40] node[below] {$?b$} (a2)
      (a2) edge[bend left = 40] node[above] {$?c$} (a3)
      (a1) edge[loop above] node[above] {$!a$} (a1)
      ;
    \end{tikzpicture}

\caption{A lossy broadcast protocol with two controller states and three user states.
}
\label{fig:rbn}
\vskip-0.6cm
\end{figure}

\begin{example}
\label{ex:rbn}
\cref{fig:rbn} depicts a  lossy broadcast protocol. 
The initial configurations are the $(c_1,\conf{v}_0)$ such that $\supp{\conf{v}_0}\subseteq\set{q_1,q_2}$.
From configuration $(c_1,\conf{v})$ with $\conf{v}=(2, 1,0)$, 
there is a step to $(c_2,\conf{v}')$ with $\conf{v}'=(1,2,0)$:
a user process takes broadcast  $q_1 \trans{!a} q_1$,  
the controller takes receive $c_1 \trans{?a} c_2$ and the other user takes  receive $q_1 \trans{?a} q_2$.
Depending on which processes receive the $a$ broadcast,
 there is also a step on $a$ from $(c_1,\conf{v})$ to $(c_1,\conf{v})$, $(c_2,\conf{v})$ and $(c_1,\conf{v}')$.
\end{example}

\smartpar{Disjunctive guard~\cite{EmersonK00}}
\emph{Disjunctive guard transitions} are of the form $(p,\existsGuard,q)$
where $p,q \in C$ or $p,q \in Q$, and $\existsGuard \subseteq C \cup Q$.
We denote the transition by $p\trans{\existsGuard} q$.
A configuration $(c,\conf{v})$ \emph{satisfies} $\existsGuard$ if $(\conf{v})(r) \ge 1$ for some $r \in \existsGuard$ or if $c \in \existsGuard$.
A \emph{disjunctive guard step} on transition $p\trans{\existsGuard} q$ is only enabled from configurations $(c,\conf{v})$ that satisfy $\existsGuard$.
Then, it consists of one or more processes taking the transition $p\trans{\existsGuard} q$ (like in internal steps), such that the resulting configuration $(c',\conf{v'})$ also satisfies $\existsGuard$, i.e., a moving process cannot be the one that satisfies the guard.
We denote the step by $(c,\conf{v}) \trans{} (c',\conf{v}')$.

\smartpar{Synchronization~\cite{BertrandDGGG19}}
Let $\Sigma$ be a finite set of labels.
\emph{Synchronization transitions} are of the form
$(p,a,q)$ with $a \in \Sigma$ 
and $p,q \in C$ or $p,q \in Q$, also denoted $p\trans{a} q$.
In a \emph{synchronization step on $a$} from a configuration  $(c,\conf{v})$,
 all processes take a transition with label $a$, if such a transition is available in their current state
(otherwise they stay in their current state).
If $(c',\conf{v}')$ is the resulting configuration, then we 
denote the step by $(c,\conf{v}) \trans{} (c',\conf{v}')$.
\begin{figure}[!ht]
    \centering
    \begin{tikzpicture}[font=\footnotesize]
        \node[location,initial] (a1) {$c_1$};
        \node[location] (a2) [right =of a1] {$c_2$};

        \path[->]
        (a1) edge[bend left = 40] node[above] {$a$} (a2)
        (a2) edge[bend left = 40] node[below] {$b$} (a1)
        (a2) edge[loop above] node[above] {$c$} (a2)
        ;
    \end{tikzpicture}
    \hspace{3em}
    \begin{tikzpicture}[font=\footnotesize]
        \node[location,initial] (a1) {$q_1$};
        \node[location,initial] (a2) [right =of a1] {$q_2$};
        \node[location] (a3) [right =of a2] {$q_3$};

        \path[->]
        (a1) edge[bend left = 40] node[above] {$a, b$} (a2)
        (a2) edge[bend left = 40] node[below] {$c$} (a1)
        (a3) edge[bend left = 40] node[below] {$b$} (a2)
        (a2) edge[bend left = 40] node[above] {$c$} (a3)
        (a1) edge[loop above] node[above] {$a$} (a1)
        ;
    \end{tikzpicture}

    \caption{A synchronization protocol with two controller states and three user states.}
    \label{fig:sync}
    \vskip-0.6cm
\end{figure}

\begin{example}
    \label{ex:sync}
    Consider the synchronization protocol depicted in \cref{fig:sync}.
    From the configuration $(c_1,\conf{v})$ with $\conf{v}=(2,1,0)$, there is a step on letter $a$ to
    $(c_2,(2,1,0))$, $(c_2,(1,2,0))$, or to $(c_2,(0,3,0))$.
    The user process initially in $q_2$  does not move because there is no $a$-transition from $q_2$.
\end{example}

\smartpar{Asynchronous shared memory (ASM)~\cite{EsparzaGM16}}
\emph{ASM transitions} are of the form
$(p,l,q)$ with $l \in \set{w(a),r(a) \ | \ a \in C} $ and $p,q \in Q$.
In systems with ASM transitions, 
we assume that a transition $(a,b)$ is available for every $a,b \in C$.
We also denote a transition $(p,l,q)$ by $p\trans{l} q$, and $(a,b)$ by $a\trans{} b$.
 Transitions with $l$ of the form $w(a)$ are \emph{write transitions},
 and transitions with $l$ of the form $r(a)$ are \emph{read transitions}.
Intuitively, the controller keeps track of the value of the shared variable, 
and the user processes can read that value or give an instruction to update the value\footnote{The restriction to a single variable is for simplicity, our results extend to multiple finite-domain variables.}.
An \emph{ASM step} from a configuration  $(c,\conf{v})$ is either a write step or a read step:
A \emph{write step} is made up of one or more user processes taking a transition $p\trans{w(a)} q$
and the controller taking transition $c \trans{} a$.
A \emph{read step} is made up of one or more user processes taking a transition $p\trans{r(a)} q$
and the controller taking transition $a \trans{} a$, ensuring that $a=c$.

\medskip

Examples for systems with disjunctive guards and ASM can be found in \cref{sec:additional-example-steps}. 
Note that all of these types of steps are between configurations of the same size.
Any finite sets $C,Q$ and any combination of transitions of the types described above 
define a set of steps $\mathcal{T}$ such that $\csystem = (C, Q, \mathcal{T})$ is a CS.
As mentioned above, one of our steps sometimes corresponds to several consecutive steps
in the classical definitions of the literature. 
More precisely,
the same broadcast transition, disjunctive transition, read transition or write transition can be taken by an arbitrary number of distinct processes in the same step.
Note that these ``accelerated'' steps do not change the reachability properties of our systems:
if $(c,\conf{v}) \trans{} (c',\conf{v}')$  is a step in our definition, 
then $(c,\conf{v}) \trans{*} (c',\conf{v}')$  is a sequence of steps in the classic definition.

\section{Reduction of Parameterized Safety Verification to the 01-CS}
\label{sec:basic}

We first prove that for $\wqo$-compatible systems, parameterized safety verification, i.e., regarding properties of finite runs, can be reduced to their 01-abstraction.
Then we show that all system types introduced in \cref{sec:transition-types} are $\wqo$-compatible, as well as some new system types.

\begin{restatable}{lemma}{LmReach}
\label{lem:reachability}
    If a given CS $\csystem$ is fully $\wqo$-compatible, then 
    there exists a run 
    $(c_0,\conf{v}_0) \rightarrow (c_1,\conf{v}_1) \rightarrow \ldots \rightarrow (c_n,\conf{v}_n)$ in $\csystem$ 
    if and only if 
    there exists a run
		$(c_0,\conf{v}_0^\alpha) \trans{} (c_1,\conf{v}_1^\alpha) \trans{} \ldots \trans{} (c_n,\conf{v}_n^\alpha)$
    in the corresponding $01$-CS $\csystemabstract$
    such that $\alpha(c_i, \conf{v}_i)=(c_i,\conf{v}_i^\alpha)$ for each $i$.
\end{restatable}

\begin{proof}
Let $\csystem = (C, Q, \mathcal{T})$ be a $\wqo$-compatible CS, 
and let $\csystemabstract$ be its $01$-CS.
Assume there exists a run $(c_0,\conf{v}_0) \trans{} \ldots \trans{} (c_n,\conf{v}_n) \trans{} (c,\conf{v})$ in $\csystem$.
Then by definition of $\csystemabstract$, for each step $(c_i,\conf{v}_i) \trans{} (c_{i+1},\conf{v}_{i+1})$ of the sequence
there exists an abstract step $(c_i,\conf{v}_i^\alpha) \trans{} (c_{i+1},\conf{v}_{i+1}^\alpha)$ in $\csystemabstract$,
where $\alpha(c_i,\conf{v}_i) = (c_i,\conf{v}_i^\alpha)$.

In the other direction, assume that there exists a run $(c_0,\conf{v}_0^\alpha) \trans{} \ldots \trans{} (c_n,\conf{v}_n^\alpha) \trans{} (c,\conf{v}^\alpha)$ in $\csystemabstract$.
We prove by induction on $n$ that there exists a run in $\csystem$ with the desired properties.\\
\emph{Base case}: $n=0$. If $(c_0,\conf{v}_0^\alpha) \trans{} (c,\conf{v}^\alpha)$ is a step in $\csystemabstract$, 
then by definition there exist
 $\conf{v}_0, \conf{v}$ with 
 $\alpha(c_0, \conf{v}_0)=(c_0,\conf{v}_0^\alpha)$ and $\alpha(c, \conf{v})=(c,\conf{v}^\alpha)$
and $(c_0,\conf{v}_0) \trans{} (c,\conf{v})$. \\
\emph{Induction step}: $n \rightarrow n+1$. Let $(c_0,\conf{v}_0^\alpha) \trans{*} (c_n,\conf{v}_n^\alpha)$ 
be a run of $\csystemabstract$ and 
$(c_n,\conf{v}_n^\alpha) \trans{} (c,\conf{v}^\alpha)$ a step in $\csystemabstract$. 
    By induction hypothesis, there exists a run $(c_0,\conf{v}_0) \trans{*} (c_n,\conf{v}_n)$ of $n$ steps in $\csystem$ 
    with $\alpha(c_0, \conf{v}_0)=(c_0,\conf{v}_0^\alpha)$ and $\alpha(c_n, \conf{v}_n)=(c_n,\conf{v}_n^\alpha)$.
    By definition of the $01$-CS, there exist $\conf{w}_n,\conf{w}$ 
    with $\alpha(c_n, \conf{w}_n)=(c_n,\conf{v}_n^\alpha)$, $\alpha(c, \conf{w})=(c,\conf{v}^\alpha)$
    and step $(c_n,\conf{w}_n) \trans{} (c,\conf{w})$ in $\csystem$.
    Configurations $\conf{w}_n$, $\conf{v}_n$ (and $\conf{v}_n^\alpha$) are equal to 0 on the same states,
    so there exists $\conf{x}_n$ such that 
     $(c,\conf{w}_n) \wqo (c,\conf{x}_n)$ and  $(c,\conf{v}_n) \wqo (c,\conf{x}_n)$.
     And therefore by \cref{lm:wqo-01-counter-system} also $\alpha(c_n, \conf{x}_n)=(c_n,\conf{v}_n^\alpha)$.
		 Then, by backward $\wqo$-compatibility of $\csystem$ we get that there is a run that reaches $(c_n,\conf{x}_n)$ in $n$ steps, and by
     forward $\wqo$-compatibility we get that there exists 
     a step $(c_n,\conf{x}_n) \trans{} (c,\conf{x})$ for some $(c,\conf{x})$ such that $\alpha(c,\conf{x}) = (c,\conf{v}^\alpha)$.
    \qed
\end{proof}

\begin{example}
Consider the lossy broadcast step in the system of \cref{fig:rbn} from $(c_2,(1,0,2))$ to $(c_1,(0,2,1))$
where the controller broadcasts $b$, one of the two processes in $q_3$ takes $q_3 \trans{?b} q_2$ and the process in $q_1$ takes  $q_1 \trans{?b} q_2$.
Configuration $(c_2,(2,0,3))$ is such that  $(c_2,(1,0,2)) \wqo (c_2,(2,0,3))$.
We describe a step from this configuration to a configuration $(c,\conf{v})$ such that $(c_1,(0,2,1)) \wqo (c,\conf{v})$:
the controller broadcasts $b$, two processes take $q_3 \trans{?b} q_2$ and two processes take  $q_1 \trans{?b} q_2$.
\end{example}

\noindent We can prove full $\wqo$-compatibility for the types of steps introduced in \cref{sec:transition-types}. 

\begin{restatable}{lemma}{LmCompatibleAll}
\label{lm:compatible-all}
    CSs induced by one of the following types of steps are fully $\wqo$-compatible: lossy broadcast, disjunctive guard, synchronization, or ASM.\footnote{Internal steps can be seen as a special case of lossy broadcast, disjunctive guard, or ASM steps.}
\end{restatable}

We give the proof for CSs induced by lossy broadcasts steps,
and relegate the other (similar) proofs to the \cref{sec:appendix-sec-compat}.

\begin{proof}[Partial]
Let $\csystem = (C, Q, \mathcal{T})$ be a CS with only lossy broadcast steps.
To prove forward $\wqo$-compatibility, assume there is a step $(c,\conf{v}) \rightarrow (c',\conf{v}')$ and $(c,\conf{v}) \wqo (d,\conf{w})$.
This step is made up of $j\ge 1$ processes taking a broadcast transition $p_0 \trans{!a} p'_0$
and $k$ processes taking receive transitions $p_1 \trans{?a} p'_1, \ldots, p_k \trans{?a} p'_k$
for some $k \ge 0$.
Since $(c,\conf{v}) \wqo (d,\conf{w})$, for every state $q\in Q$, $\conf{v}(q) \le \conf{w}(q)$ and $c=d$.
Therefore a step with the same $j+k$ transitions can be taken from $(d,\conf{w})$.
Call $(d'',\conf{w''})$ the resulting configuration. 
We want to check $(c',\conf{v'}) \wqo (d'',\conf{w''})$ or modify the step from $(d,\conf{w})$ to make it true.
This means we have to satisfy the following three conditions:
\begin{enumerate}[topsep=1pt,parsep=1pt]
\item 
$d''=c'$: either $c=c'$, in which case no transition is taken by the controller in either step and $d=d''=c=c'$,
or $c \neq c'$, in which case $c=p_i, c'=p_i'$ for some $i \in \set{0,...,k}$ 
and this transition is taken in both steps, so $d''=p_i'=c'$.
\item 
$\conf{w''}(q) \ge \conf{v'}(q)$ for all  $q\in Q$: the same transitions are taken from $\conf{w}\ge \conf{v}$.
\item
$\conf{w''}(q)=0 $ if and only if $\conf{v'}(q)=0$ for all  $q\in Q$: 
if there are no such states then we are done; otherwise
suppose $\conf{v'}(q)=0$ and $\conf{w''}(q)> 0$.
This entails $\conf{w}(q)> 0$ and thus also $\conf{v}(q)> 0$ by definition of $\wqo$.
Informally, this means state $q$ was emptied in the step $(c,\conf{v}) \rightarrow (c',\conf{v}')$;
one of the transitions taken is of the form $q=p_i \trans{a \star} p'_i$ with $\star \in \set{!,?}$.
We modify the step from $(d,\conf{w})$
by adding $\conf{w}(q)-\conf{v}(q)$ iterations of $p_i \trans{\star a} p'_i$, i.e.,  enough to empty $q$.
The resulting configuration $(d',\conf{w'})$ is such that  $(c',\conf{v'}) \wqo (d',\conf{w'})$. 
\end{enumerate}
Backward $\wqo$-compatibility can be proven in a similar way.
\qed
\end{proof}

As a new communication primitive, we can extend synchronization transitions (as introduced in \cref{sec:transition-types}) to \emph{guarded synchronizations}, which are additionally labeled with a pair $(\existsGuard,\forallGuard)$ with $\existsGuard, \forallGuard \subseteq (C \cup Q)$, and then denoted $p \trans{a, (\existsGuard,\forallGuard)} q$.
The step is defined as for synchronization steps, except that a \emph{synchronization step guarded by $(\existsGuard, \forallGuard)$} is only enabled from a configuration $\config$ with $S = \supp{\conf{v}} \cup \{c\}$ if $S \cap \existsGuard \neq \emptyset$ (there exists a $g \in S$ with $g \in \existsGuard$), and $S \subseteq \forallGuard$ (for all $g \in S$ we have $g \in \forallGuard$). 
That is, they are interpreted as a disjunctive and a conjunctive guard, respectively, and we can mix both types of guards, even in the same transition.

\begin{example}[Guarded Synchronization]
Consider the synchronization protocol from \cref{ex:sync}, assuming that the action on $a$ is guarded by $\existsGuard = \{c_1, \loc_1\}$ and $\allGuard = Q \setminus \{ \loc_3 \}$.
Then it is enabled from $(c_1,\conf{v})$ with $\conf{v}=(2,1,0)$ (as there is a process in $c_1 \in \existsGuard$, and no processes are in $\loc_3$).
It is however not enabled from $(c_2,\conf{w})$ with $\conf{w}=(0,2,0)$ (as there is no process in a state from $\existsGuard)$, and also not from $(c_1,\conf{w}')$ with $\conf{w}'=(0,2,1)$ (as one process is in $\loc_3 \notin \allGuard$).
\end{example}

CSs with guarded synchronizations are also fully $\wqo$-compatible.
\begin{restatable}{lemma}{LmCompatibleGSync}
  \label{lm:compatible-guard-synchronizations}
      CSs induced by guarded synchronization steps are fully $\wqo$-compatible. 
\end{restatable}

\begin{proof}
    To prove forward $\wqo$-compatibility, assume there is a step $\config \trans{\act, (\existsGuard,\allGuard)} \configPrime$ and $\config \wqo (d,\conf{w})$.
    Note that, as the step is enabled from $\config$, there exists a state $\loc_\exists \in \existsGuard$ that is also in $\supp{\conf{v}} \cup \{c\}$, and every state in $\supp{\conf{v}} \cup \{c\}$ is also in $\allGuard$. 
		By definition of $\wqo$ it follows that for $(d,\conf{w})$ there is at least one process in $\loc_\exists$ and all states in $(C \cup \Loc) \setminus \allGuard$ remain empty. 
		Consequently, the synchronization on $\act$ is also enabled in $(d,\conf{w})$ and forward compatibility follows by forward compatibility of synchronization actions.

    To prove backward $\wqo$-compatibility, assume there is a step $\config \trans{\act, (\existsGuard,\allGuard)} \configPrime$ and $\configPrime \wqo (d',\conf{w}')$. 
		By backward compatibility of synchronization steps, we know that there must exist a configuration $ (d,\conf{w})$ such that $\config \wqo (d,\conf{w})$. 
		By the same reasoning as above the configuration does also satisfy both guards $\existsGuard$ and $\allGuard$.
    \qed
\end{proof}

This result can be considered surprising, as the combination of disjunctive and conjunctive guards for internal transitions leads to undecidability~\cite{EmersonK03}. It is key that we use synchronizations and not internal transitions here.

However, note that in each of the compatibility proofs, it is enough to prove $\wqo$-compatibility for a single (arbitrary) step of the system.
Therefore, we can also mix different types of steps in the same CS.

\begin{theorem}
\label{cor:combinations}
{A CS is fully $\wqo$-compatible if its steps can be partitioned into sets such that $\wqo$-compatibility holds for each set.
In particular, a CS is fully $\wqo$-compatible  if each of its steps is induced by one of the following transition types: internal, disjunctive guard, lossy broadcast, (guarded) synchronization, or ASM.}
\end{theorem}

\begin{remark}
Note that \cref{cor:combinations} does not make a statement about transitions that combine the characteristics of different types of transitions.
Nonetheless, compatibility with $\wqo$ holds for many extensions of the types of steps we consider.
 In particular, all of them can be extended with disjunctive guards, and even with conjunctions of disjunctive guards, i.e., requiring multiple disjunctive guards to be satisfied at the same time (as in \cite{JacobsS18}).
 Moreover, shared finite-domain steps can have several shared variables encoded into the controller states.
\end{remark}

Compatibility with $\wqo$ is related to what is sometimes called the ``copycat property''.
Informally, this property holds if,
whenever a process can move from $p$ to $p'$ in a step $(c,\conf{v}) \trans{} (c',\conf{v}')$, 
then any additional processes that are in $p$ in a configuration $(c, \conf{v}+i \cdot \conf{p})$ can also move to $p'$ in a sequence of steps $(c,\conf{v}+i \cdot \conf{p}) \trans{*} (c, \conf{v}'+i \cdot \conf{p'})$,
``copying'' the movement of the first process.
We use this property implicitly to prove $\wqo$-compatibility, 
and prove or reprove it for all the systems considered here.

\section{Parameterized Reachability Problems}
\label{sec:reach}

We define the type of parameterized problems we consider and
show that we can solve them in polynomial space using the 01-CS.
Then, we introduce a class of $\wqo$-compatible CSs that have not been considered in the literature before, and use them to prove \PSPACE-hardness of any of these problems.

\subsection{The Cardinality Reachability Problem}

Inspired by Delzanno et al.~\cite{DelzannoSTZ12}, we define a \emph{cardinality constraint} $\varphi$ as a formula in the following grammar, where $c \in C$, $a \in \N$,  and $q \in Q$:
\[ \varphi ::= \control = c \mid \control \neq c \mid \#q \geq a \mid \#q = 0  \mid \varphi \land \varphi \mid \varphi \lor \varphi \]

The satisfaction of cardinality constraints is defined in the natural way, e.g., $(c,\conf{v}) \models \control = c'$ if $c =c'$, and $(c,\conf{v}) \models \#q \geq a$ if $\conf{v}(q) \geq a$.
In~\cite{DelzannoSTZ12}, there are no atomic propositions $\control = c$ nor $\control \neq c$ (since they do not have a controller process), but there is $\#q \leq b$ for any $b \in \N$ (which is not supported in the $01$-abstraction, except for the special case $\#q = 0$).

Given a CS $\csystem$ and a cardinality constraint $\varphi$,
the \emph{cardinality reachability problem (CRP)} asks whether 
a configuration $(c,\conf{v})$ with $(c,\conf{v}) \models \varphi$ is reachable in $\csystem$.
\begin{itemize}
\item
Let $CC[\geq a]$ be the class of cardinality constraints in which atomic propositions 
are only of the form $\#q \geq a$  for any $a \in \N$.

\item
Let $CC[\geq a, =0]$ be the class of cardinality constraints in which atomic propositions 
are only of the form $\#q = 0$ or $\#q \geq a$  for any $a \in \N$.

\item
Let $CC[\control, \geq a, =0]$ be the class of cardinality constraints in which atomic propositions 
are of the form $\control = c, \control \neq c, \#q = 0$ or  $\#q \geq a$  for any $a \in \N$, i.e., the maximal class.
\end{itemize}
For a given $\varphi \in CC[\control, \geq a, =0]$, let $\varphi_\alpha = \varphi[\#q \geq a \mapsto \#q \geq 1]_{a \in \N^+}$, i.e., the result of replacing every atomic proposition of the form $\#q \geq a$ with the proposition $\#q \geq 1$ if $a  \in \N^+$.
We write \emph{CRP for} $S$ to denote that we consider the CRP problem for a cardinality constraint in 
$S \in \{CC[\geq a], CC[ \geq a, =0], CC[\control,\geq a,=0]\}$.

Many parameterized reachability problems can be expressed in CRP format, e.g., coverability, control-state reachability, or the target problem~\cite{DelzannoSZ10} (see \cref{ex:crp-problems} in \cref{app:reach} for details).

\subsection{Deciding the CRP for $\wqo$-compatible Counter Systems}

We show that CRP is \PSPACE-complete for CSs given a light restriction. We start by showing that 
checking CRP in a fully  $\wqo$-compatible CS can be 
reduced to checking the 01-CS.

\begin{lemma}
\label{lm:reduce}
Let $\csystem$ be a fully $\wqo$-compatible CS, $\csystemabstract$ its 01-CS and let $\varphi \in CC[\control,\geq a,=0]$.
\begin{enumerate}
\item %
If a configuration $(c,\conf{v})$ that satisfies $\varphi$ is reachable in $\csystem$, then $(c,\absconf{v})=\alpha(c,\conf{v})$ satisfies $\varphi_\alpha$ and is reachable in $\csystemabstract$. 
\item %
If an abstract configuration $(c,\absconf{v})$ that satisfies $\varphi_\alpha$ is reachable in $\csystemabstract$, then there exists $(c,\conf{v})$ that satisfies $\varphi$, is reachable in $\csystem$, and with $\alpha(c,\conf{v}){=}(c,\absconf{v})$.
\end{enumerate}
\end{lemma}
\begin{proof}
To prove (1), assume there is a configuration $(c,\conf{v})\models \varphi$ that is reachable in $\csystem$.
Then $(c,\absconf{v})=\alpha(c,\conf{v})$ 
 is reachable in $\csystemabstract$
 by \cref{lem:reachability}, and the fact that $\alpha$ of an initial configuration in $\csystem$
 is an initial configuration in $\csystemabstract$.
 It is easy to see that $(c,\conf{v}) \models \varphi$ entails $(c,\absconf{v}) \models \varphi_\alpha$ by definition.
To prove (2), assume there is a configuration $(c,\absconf{v}) \models \varphi_\alpha$ 
 reachable in $\csystemabstract$.
 By \cref{lm:wqo-01-counter-system}, 
 any $(c,\conf{v}) \succeq_0 (c,\absconf{v})$ maps to $(c,\absconf{v})$ by $\alpha$. 
 Choose $\conf{v}$ such that $\conf{v}(q)=0$ if $\absconf{v}(q)=0$
 and $\conf{v}(q)=A$ otherwise, where $A$ is the highest lower bound in $\varphi$,
 that is $A \ge a$ for all $\#q \geq a$ appearing in  $\varphi$.
 Then $(c,\conf{v}) \models \varphi$, and by \cref{lem:reachability},  $(c,\conf{v})$ is reachable in $\csystem$. 
\qed
\end{proof}

Let $\csystem=(C,Q,\mathcal{T})$ be a $\wqo$-compatible CS.
A product of wqos is a wqo \cite{Kruskal72}, so 
 $\wqo \times \wqo$ is a wqo on $(C \times \N^Q)^2$.
Given a wqo $\preceq$ on a set $S$, it is the case that for
 every subset $X \subseteq S$ there exists a finite subset $Y \subseteq X$ of minimal elements  
such that for every $x \in X$ there exists $y \in Y$ with $y \preceq x$ \cite[Thm.~1.1]{LucaV94}.
This subset is called the \emph{finite basis} of $X$, 
and it is unique if the wqo is antisymmetric (our $\wqo$ is antisymmetric).
Applying this to the step relation $\mathcal{T}$ of $\csystem$ and the wqo $\wqo \times \wqo$ implies the existence of 
 a finite basis $Y$ of $\mathcal{T}$, since $\mathcal{T}\subseteq (C \times \N^Q)^2$.
 Then define 
\begin{center}
 $B_\csystem = max(\conf{v}(q) \mid ((c,\conf{v}),(c',\conf{v}')) \in Y, q \in Q )$,
\end{center}
i.e., the maximal number of user processes per state
in any step in the basis $Y$.

\begin{remark}
The constant $B_\csystem$ is usually small in counter systems. 
For example, for $\csystem$ a counter system with only lossy broadcast steps, $B_\csystem$ is bounded by $|Q|$:
a step depends on one broadcast transition and an arbitrary number of receive transitions.
In the worst case, a minimal step is such that, for a given state $p$, the broadcast is $p \trans{!a} p'$
and there are receive transitions $p \trans{?a} q$ for every $q \in Q \setminus \{ p' \}$.

For disjunctive guards, $B_\csystem \le 2$;
for synchronizations, $B_\csystem \le |Q|$; and
for ASM, $B_\csystem \le 1$.
For a CS with several types of these steps, 
$B_\csystem$ is bounded by the maximum of the $B_\csystem$ given here.
See \cref{rmk:bc} in \cref{app:reach} for details.
\end{remark}

{We say that a fully $\wqo$-compatible CS $\csystem = \csystemdef$ is \emph{polynomially abstractable} if $B_\csystem$ is polynomial in $|C|$ and $|Q|$, and membership in $\mathcal{T}$ can be checked in \Polytime.
All the types of systems that we have considered so far are polynomially abstractable.}

\begin{theorem}
\label{thm:reach-unrestricted}

    Let $\csystem$ be a polynomially abstractable CS for which $B_\csystem$ is known.
    Then the CRP for $\csystem$ and $\varphi \in CC[\control,\geq a, =0]$ is in \PSPACE.

\end{theorem}
\begin{proof}
Let $\csystem=(C,Q,\mathcal{T})$ be a CS that is $\wqo$- compatible and polynomially abstractable for a known $B_\csystem$,
let $\csystemabstract$ be its 01-CS and let $\varphi \in CC[\control,\geq a, =0]$.
By \cref{lm:reduce}, it suffices to check whether there exists  
an abstract configuration $(c,\absconf{v})$ that satisfies $\varphi_\alpha$ 
and that is reachable in $\csystemabstract$.
There are at most $|C| \cdot 2^{|Q|}$ abstract configurations.
We explore the abstract system $\csystemabstract$ non-deterministically,
guessing an initial configuration, then a path from this configuration.
At each step, we check if the current configuration $(c,\absconf{v})$ 
satisfies $\varphi_\alpha$ (this can be done in polynomial time in the number of states).
If it does not, we guess a step $(c,\absconf{v}) \trans{} (c',\absconf{v'})$ in $\csystemabstract$.
To do this, we guess a configuration $(c,\conf{v})$ of $\csystem$
with $1 \le \conf{v}(q) \le B_\csystem$ for all $q$ such that $\absconf{v}(q)=1$, 
and with $\conf{v}(q) =0$ elsewhere. 
We guess a configuration $(c',\conf{v}')$ of the same size as $(c,\conf{v})$.
We check if $(c,\conf{v}) \rightarrow (c',\conf{v}')$ is a step of $\mathcal{T}$ (which by assumption can be done in polynomial time).
If it is, then $\alpha(c,\conf{v}) \rightarrow \alpha(c',\conf{v}')$ is a step in $\csystemabstract$.

Let $Y$ be the finite basis of $\mathcal{T}$.
It is enough to check whether there exists a step with only counters under $B_\csystem$ 
because 
$(c,\absconf{v}) \rightarrow (c',\absconf{v'})$ is a step in $\csystemabstract$
if and only if there exists a step $(c,\conf{v}) \rightarrow (c',\conf{v}')\in Y$
with $\alpha(c,\conf{v}) = (c,\absconf{v})$ 
and $ \alpha(c',\conf{v}') = (c',\absconf{v'})$.
Indeed, $(c,\conf{v}) \rightarrow (c',\conf{v}')\in Y$ implies an abstract step because
$Y \subseteq \mathcal{T}$.
In the other direction, if $(c,\absconf{v}) \rightarrow (c',\absconf{v'})$ is a step in $\csystemabstract$ 
then there exists  $(d,\conf{w}) \rightarrow (d',\conf{w}')\in \mathcal{T}$
with $\alpha(d,\conf{w}) = (c,\absconf{v})$ 
and $ \alpha(d',\conf{w}') = (c',\absconf{v'})$.
By definition of $Y$ there exists $(c,\conf{v}) \rightarrow (c',\conf{v}')\in Y$ with
$(c,\conf{v}) \wqo (d,\conf{w})$ and $(c',\conf{v}') \wqo (d',\conf{w}')$.
By \cref{lm:wqo-01-counter-system}, this entails 
$\alpha(c,\conf{v}) = (c,\absconf{v})$ 
and $ \alpha(c',\conf{v}') = (c',\absconf{v'})$.
This procedure is in polynomial space in the number of states and $B_\csystem$ because  
each configuration can be written in polynomial space, 
 all checks can be performed in polynomial time, and
 by Savitch's Theorem $\PSPACE = \NPSPACE$ so we can give a non-deterministic algorithm. 
\qed
\end{proof}

\smartpar{\PSPACE-Hardness of the CRP}
\label{sec:hardness}
Our upper bound on the complexity of CRP for $\wqo$-compatible CSs is higher than some of the existing complexity results for systems that fall into this class\footnote{E.g., for RBN without a controller, CRP for $CC[\geq 1]$ is in \Polytime, and for $CC[\geq 1,=0]$ it is in \NP~\cite{DelzannoSZ12}.}.
We show that this complexity is unavoidable, implying that the class of fully $\wqo$-compatible systems is more expressive than its instances considered in the literature.

We prove \PSPACE-hardness by a reduction of the intersection non-emptiness problem for deterministic finite automata (DFA) \cite{conf/focs/Kozen77} to the CRP. The detailed construction can be found in \cref{app:hardness}. The idea is to view the DFA  as systems communicating via synchronization transitions, where the set of actions is the input alphabet.
The intersection of the languages accepted by the automata is then non-empty iff some configuration is reachable such that in each automaton there is at least one accepting state covered by a process.
This constraint can be encoded into a constraint $\varphi \in CC[\geq 1]$ and the construction does not use a controller, therefore deciding the CRP even in this restricted setting is \PSPACE-hard.
As a consequence, we get \PSPACE-completeness for the CRP of $\wqo$-compatible systems.

\begin{theorem}
  \label{thm:CRP-PSPACE-complete}
    Let $\csystem$ be a polynomially abstractable CS for which $B_\csystem$ is known.
    Then the CRP for $\csystem$ and $\varphi \in CC[\control,\geq a, =0]$ is \PSPACE-complete.
\end{theorem}\nointerlineskip

\section{Parameterized Model Checking of Trace Properties}
\label{sec:trace}

A large part of the parameterized verification literature has focused on model checking of stutter-insensitive trace properties of a single process, or a fixed number $k$ of processes~\cite{EmersonK00,EmersonK03,AminofKRSV18,AusserlechnerJK16,JacobsS18}.
We sketch how our framework improves existing results in this area, including for liveness properties.

\smartpar{Trace Properties}
Given a $CS$ $\csystem = \csystemdef$, a \emph{trace} of the controller is a finite word $w \in C^*$ obtained from a run $\rho$ of $\csystem$ by projection on the first element of each configuration, and removing duplicate adjacent letters.
We denote by $\traces{\csystem}$ the set of all finite traces that can be obtained from runs of $\csystem$, and by $\tracesinf{\csystem}$ the set of \emph{infinite} traces.
Define similarly $\traces{\csystemabstract}$ and $\tracesinf{\csystemabstract}$ for the $01$-CS.
A \emph{safety property} $\varphi$ is a prefix-closed subset of $C^*$.
We say that $\csystem$ \emph{satisfies} the safety property $\varphi$, denoted $\csystem \models \varphi$, if $\traces{\csystem} \subseteq \varphi$.

\smartpar{Existing Results}
Many of the existing results for deciding trace properties are based on cutoffs~\cite{EmersonK00,EmersonK03,AusserlechnerJK16,JacobsS18}.
That is, they view the system as a parallel composition $A{\parallel}B^n$ of controller and user processes, and derive a \emph{cutoff} for $n$, i.e., a number $c$ such that $A{\parallel}B^c \models \varphi \iff \forall n \geq c: A{\parallel}B^n \models \varphi$.
This reduces the problem to a (decidable) model checking problem over a finite-state system.
However, since the cutoff $c$ is usually linear in $\card{B}$, the state space of this finite system is in the order of $O\left(\card{A} \times \card{B}^{\card{B}}\right)$.

As an improvement of these results, Aminof et al.~\cite{AminofKRSV18} have shown that $\tracesinf{A}$ can be recognized by a B\"uchi-automaton of size $O(\card{A}^2 \cdot 2^{\card{B}})$, and the same for $\tracesinf{B}$, the infinite traces of a single user process in the parameterized system.

\smartpar{Deciding Trace Properties in the $01$-CS}
As a direct consequence of \cref{lem:reachability} we get:

\begin{lemma}
\label{lem:PMCP}
If CS $\csystem$ is fully $\wqo$-compatible and $\csystemabstract$ is its $01$-CS, then $\traces{\csystem}=\traces{\csystemabstract}$.
\end{lemma}

Note that the size of $\csystemabstract$ is $\card{C} \cdot 2^{\card{Q}}$, i.e., smaller than the B\"uchi automaton in the result of Aminof et al.~\cite{AminofKRSV18}.
On the other hand, our result in general only holds for finite traces.
To see that $01$-abstraction is not precise for infinite traces, consider the following example.

\begin{example}
\begin{figure}[h!]
    \centering
        \begin{tikzpicture}
            \node [location, initial, initial text=] (qinit) at (0,0) {$\initloc$};
            \node [location] (q1) at (2,0) {$\loc_1$};

            \path[->]
                (qinit) edge [] node [above] {$!a$} (q1)
                (qinit) edge [loop above] node [above] {$?a$} (qinit)
            ;
        \end{tikzpicture}
        \hspace*{2em}
        \begin{tikzpicture}
            \node [] (c0) at (0,0) {$(1,0)$};
            \node [] (c1) at (2,0) {$(1,1)$};

            \path[->] 
                (c0) edge [] node [above] {$a$} (c1)
                (c1) edge [loop above] node [] {${a}^\omega$} (c1)
            ;
        \end{tikzpicture}

    \caption{Example system with spurious loop}
    \label[figure]{fig:spurious}
\end{figure}%
    Consider the CS $\csystem$ based on lossy broadcast depicted in~\cref{fig:spurious}.
		To its right we depict an infinite run of its $01$-CS that executes a lossy broadcast on $a$ infinitely often: on the first $a$, it moves from $\conf{v}= (1,0)$ to $\conf{v'}=(1,1)$, and then any further application loops in $\conf{v'}=(1,1)$.
		However, such a behavior is not possible in $\csystem$: any concrete run of $\csystem$ will start with a fixed number $n$ of processes, and therefore has to stop after $n$ steps.
\end{example}

Despite this, we can extend \cref{lem:PMCP} to infinite traces for \disjunctive{}s:

\begin{restatable}{lemma}{LmPMCPinfinite}
\label{lem:PMCP-infinite}
If $\csystem$ is a $\wqo$-compatible CS induced by disjunctive guard transitions, then $\tracesinf{\csystem}=\tracesinf{\csystemabstract}$.
\end{restatable}

The proof for the lemma can be found in \cref{app:sec-trace}.

Note that it is easy to obtain a B\"uchi automaton $B$ that recognizes the same language as $\csystemabstract'$: 
The states of $B$ are the configurations of $\csystemabstract$, plus a special sink state $\bot$. 
Labels of transitions in $B$ are from the set of minimal steps $\Dmin$.
There is a transition between two automaton states with label $D \in \Dmin$ if both are configurations and there is a transition based on $D$ between them in $\csystemabstract'$, and between a configuration and $\bot$ there is a transition labeled $D$ if there is no transition based on $D$ and starting in this configuration in $\csystemabstract'$.
Finally, there is a self-loop with all labels from $\Dmin$ on $\bot$, and every state except $\bot$ is accepting.

Also note that we get corresponding results to \cref{lem:PMCP} and \cref{lem:PMCP-infinite} for traces of a user process,
by encoding one copy of the user process into the controller (i.e., the controller simulates the product of the original controller and one user process), such that we can directly observe the traces of one fixed user process.
The same construction works for any fixed number $k$ of user processes.
\cref{tab:PMCP} summarizes our results on trace properties, and compares them to existing results from the literature.

\begin{table*}[b]
\caption{Decidability and Complexity of PMCP over finite and infinite traces, Comparison of Our Results to Existing Results} 
\label{tab:PMCP}
\resizebox{\textwidth}{!}{
\begin{tabular}{l|l|l||l|l|l}
\toprule
\multicolumn{3}{c}{Our Results} & \multicolumn{3}{c}{Existing Results}\\
System Class & Traces & Result & System Class & Traces & Results\\
\midrule
\textbf{$\mathbf{\wqo}$-compatible systems} & finite & $\traces{\csystem}=\traces{\csystemabstract}$ &  \disjunctive{}s & finite & $\traces{\csystem} = \traces{B}$~\cite{AminofKRSV18}\\
& & \textbf{(where $\mathbf{|\csystemabstract|=\card{C} \cdot 2^{\card{Q}}}$)} & & & (where $\card{B}=\card{C}^2\cdot 2^{\card{Q}}$)\\
\midrule
disjunctive systems & infinite & $\tracesinf{\csystem}=\tracesinf{\csystemabstract}$ & disjunctive systems & infinite & $\traces{\csystem} = \traces{B}$~\cite{AminofKRSV18}\\
& & \textbf{(where $\mathbf{|\csystemabstract|=\card{C} \cdot 2^{\card{Q}}}$)} & & & (where $\card{B}=\card{C}^2\cdot 2^{\card{Q}}$)\\
\bottomrule
\end{tabular}
} \vskip-0.2cm
\end{table*}

\smartpar{Automata-based Model Checking}
\cref{lem:PMCP} and \cref{lem:PMCP-infinite} state language equivalences, but do not directly solve the PMCP.
We assume that the specification $\varphi$ is given in the form of an automaton $\mathcal{A}_\varphi$ that accepts the language $\varphi$.
By \cref{lem:PMCP}, for safety properties it is then enough to check whether the product $\csystemabstract \times \mathcal{A}_\varphi$ 
can reach a state which is non-accepting for $\mathcal{A}_\varphi$, and similarly %
for the PMCP over infinite traces based on \cref{lem:PMCP-infinite}.

\section{Transition Counter Systems}
\label{sec:small-tcs}

In this section, we give a restriction on $\wqo$-compatible CSs,
and show that CRP for $CC[\geq a]$ and  $CC[\geq a, =0]$ is \Polytime{}- and \NP{}-complete respectively.
This restriction is inspired by Delzanno et al.~\cite{DelzannoSTZ12}, who 
study reconfigurable broadcast networks (RBN) without a controller process.
Accordingly, they consider the CRP for cardinality constraints without the propositions $\control = c, \control \neq c$.
They show that for RBN, CRP for $CC[\geq 1]$ is \Polytime-complete and CRP for $CC[\geq 1, =0]$ is \NP-complete, where
$CC[\geq 1]$ are the cardinality constraints in which atomic propositions are only of the form $\#q \geq 1$.

\smartpar{Transition Counter Systems}
We consider CSs
in which steps are based on local transitions between states,
as is the case for the system models we have studied in this paper. 
Here, we do not consider a controller process, i.e., configurations are in $\N^Q$.

A CS \emph{without controller} is
$\csystem = (Q, \mathcal{T})$, where
$Q$ is a finite set of states, 
 the step relation is $\mathcal{T} \subseteq \N^Q \times  \N^Q$, and
configurations are $\conf{v} \in \N^Q$.
The results for CSs with controller in the previous sections still hold 
for CSs without controller:
given $\csystem = (Q, \mathcal{T})$  without controller, add to it a
set $C=\set{c}$ and consider configurations in which one process is in $c$.
Since no steps of $\mathcal{T}$ involve $C$, this process cannot move and can be ignored.

Fix $\csystem = (Q, \mathcal{T})$ a CS without controller,
and $\delta \subseteq Q^2$ a set of transitions between states 
(the transitions may have labels, but we ignore these for now).
We  denote transitions $(p,p')$ by $p \trans{} p'$.
Given a multiset of transitions $D \in \N^\delta$,
 let $\pre{D}$ be the multiset of states $p$ such that 
 $\pre{D}(p)=m$ if there are $m$ transitions in $D$
 of the form  $p \trans{} p'$ for some $p'$.
Let $\post{D}$ be the multiset of states $p'$ such that 
 $\post{D}(p')=m$ if there are $m$ transitions in $D$
 of the form  $p \trans{} p'$ for some $p$.

Let $\conf{v},\conf{v}'$ be two configurations of $\N^Q$, and   
let $D \in \N^\delta$ be a multiset of transitions.
We say $\conf{v}'$ is obtained by \emph{applying} $D$ to $\conf{v}$ if
 $\conf{v}' = \conf{v} - \sum_{i=1}^k \conf{p_i} + \sum_{i=1}^k \conf{p_i'},$
where $(p_1,p_1'), \ldots, (p_k,p_k')$ are the transitions of $D$ enumerated with multiplicity.
Note that the result is only well-defined if $\conf{v}(p)\ge \pre D (p)$ for all $p \in Q$,
and $\supp{\pre D} \subseteq \supp{\conf{v}}$ (recall that $\supp{\conf{m}}$ is the support of a multiset $\conf{m}$),
and that these conditions are ensured by our definitions of steps in \cref{sec:transition-types}. 

A transition counter system 
is characterized by a finite set $\Dmin$ of ``minimal steps'',
where a $D\in \Dmin$ is a multiset of transitions of $\delta$
such that each transition appears at most once in $D$, 
i.e., $D \in \set{0,1}^\delta$.
Intuitively,  $D$ is a group of transitions that must be taken together in a step, 
and this group is of minimal size.
All steps of a transition CS are based on a  $D\in\Dmin$, by applying each transition of $D$ one or more times.
 
 \begin{definition}
 \label{def:transition-cs}
A CS without a controller 
$\csystem = (Q, \mathcal{T})$ is a  \emph{transition counter system (TCS)}
if there exists a finite set of transitions $\delta\subseteq Q^2$ and
a finite set of minimal steps $\Dmin \subseteq \set{0,1}^\delta$ such that
$\conf{v} \trans{} \conf{v}'$ is a step of $\csystem$ if and only if
$\conf{v}'$ is obtained by applying $D\in \N^\delta$ to $\conf{v}$,
where $D$ is a multiset of local transitions such that there exists a $D_0 \in \Dmin$
with $\supp{D}=\supp{D_0}$, i.e. $D$ and $D_0$ are non-zero on the same transitions.
\end{definition}
Notice that TCSs are entirely defined by the tuple $(Q, \delta, \Dmin)$, and
they are always polynomially abstractable: testing membership of a step in $\mathcal{T}$ is always in \Polytime, and $B_\mathcal{C} \leq \card{Q}$.

\begin{restatable}{lemma}{LmTransitionSystems}
\label{lm:transition-systems}
\begin{enumerate}
\item CSs without controller and with only lossy broadcast steps, or only disjunctive guard steps,
 are TCSs.
\item CSs without controller and with only synchronization steps are not TCSs.
\end{enumerate}
\end{restatable}

\begin{proof}
To prove (1), first consider a counter system $\csystem= (Q, \mathcal{T})$ without controller, and
with only lossy broadcast steps based on broadcast and receive transitions from a set $\delta$.
These are still well defined, as the definition did not distinguish between controller and user processes.
Define $\Dmin$ to be the set of $D \in \set{0,1}^\delta$ such that for each
 broadcast transition $t_0=p_0 \trans{!a} p'_0 \in \delta$
 and each subset of receive transitions $t_1=p_1 \trans{?a} p'_1, \ldots, t_k=p_k \trans{?a} p'_k \in \delta$ for the same letter $a$,
there is a $D= \conf{t}_0 + \conf{t}_1 + \ldots \conf{t}_k $.
Then $\csystem$ is equivalent to the transition counter system $\mathcal{D}=(Q, \delta, \Dmin)$
in the following sense:
there is a step $\conf{v} \trans{} \conf{w}$ in $\csystem$ if and only if
there is a step $\conf{v} \trans{} \conf{w}$ in $\mathcal{D}$.

Now, consider a counter system $\csystem= (Q, \mathcal{T})$ without controller, and
with only disjunctive guard steps based on transitions from a set $\delta$.
These are still well defined, as the definition did not distinguish between controller and user processes.
Define $\Dmin$ to be the set of $D \in \set{0,1}^\delta$
such that  for each pair of
transitions $t= p \trans{\existsGuard} q \in \delta$ and $r \trans{} r$ for $r \in \existsGuard$,
there is a $D= \conf{t} + \conf{r}$.
Then $\csystem$ is equivalent to the transition counter system $\mathcal{D}=(Q, \delta, \Dmin)$,
in the same sense as above.

Regarding (2), consider a counter system $\csystem= (Q, \mathcal{T})$ without controller, and
with only synchronization steps based on transitions from a set $\delta$.
These are still well defined, as the definition did not distinguish between controller and user processes.
Assume there exists $\Dmin$ such that
there is a step $(c,\conf{v}) \trans{} (d,\conf{w})$ in $\csystem$ if and only if
there is a step $(c,\conf{v}) \trans{} (d,\conf{w})$ in the transition counter system $\mathcal{D}=(Q, \delta, \Dmin)$.
Assume there is a $D\in \Dmin$ containing a transition $p\trans{a} q$ and no transition $p\trans{a} p$.
Consider a configuration $\conf{v}$ such that $\conf{v}(p) = \pre D (p)$ for all $p \in Q$.
Applying $D$ to $\conf{v}$ defines a step $\conf{v} \trans{} \conf{v}'$ in $\mathcal{D}$.
Now consider configuration $\conf{v}'' = \conf{v} + \conf{p}$.
By definition of a  transition counter system,
 applying $D$ to $\conf{v}''$ also defines a step in $\mathcal{D}$.
However, this is not a step of $\csystem$ because there is a process of $\conf{v}''$ in $p$
which takes no transition in the step.
This is not possible in a synchronization step, where all processes in states with an $a$-labeled transition
must take an $a$-labeled transition.
Therefore counter systems without controller with only synchronization steps
may not be equivalent to transition counter systems.
\qed
\end{proof}

It is known that RBN can simulate ASM systems~\cite{Gandalf21}, so they can indirectly be modeled as TCSs.
We now show that TCSs are $\wqo$-compatible by design.

\begin{restatable}{lemma}{LmTCS}
\label{lm:tcs}
TCSs are fully $\wqo$-compatible.
\end{restatable}

\begin{proof}
Let $\csystem$ be a transition counter system given by $(Q, \delta, \Dmin)$.
To prove forward $\wqo$-compatibility, assume there is a step $\conf{v}\rightarrow \conf{v}'$ and $\conf{v} \wqo \conf{w}$.
There exists a multiset of transitions $D$ such that $ \conf{v}'$ is obtained by applying $D$ to $\conf{v}$.

For every state $q\in Q$, $\conf{v}(q) \le \conf{w}(q)$.
Therefore $D$ can be applied to $\conf{w}$.
Call $\conf{w''}$ the resulting configuration. 
We want to check $\conf{v'} \wqo \conf{w''}$ or modify the step from $\conf{w}$ to make it true.
This means we have to satisfy the following conditions:
\begin{enumerate}[topsep=1pt,parsep=1pt]
\item 
$\conf{w''}(q) \ge \conf{v'}(q)$ for all  $q\in Q$: the same transitions are taken from $\conf{w}\ge \conf{v}$, so this will hold.
\item
$\conf{w''}(q)=0 $ if and only if $\conf{v'}(q)=0$ for all  $q\in Q$: 
if there are no such states then we are done; otherwise
suppose $\conf{v'}(q)=0$ and $\conf{w''}(q)> 0$.
This entails $\conf{w}(q)> 0$ and thus also $\conf{v}(q)> 0$ by definition of $\wqo$.
This means state $q$ was emptied in the step $\conf{v} \rightarrow \conf{v}'$;
one of the transitions in $D$ is of the form $q \trans{} p$.
We call $D'$ the multiset of transitions obtained 
by adding $\conf{w}(q)-\conf{v}(q)$ iterations of $q \trans{} p$ to $D$, i.e.,  enough to empty $q$.
The configuration $\conf{w'}$ obtained by applying $D'$ to $\conf{w}$ is such that  $\conf{v'} \wqo \conf{w'}$. 
\end{enumerate}
Backward $\wqo$-compatibility can be proven in a similar way.
\qed
\end{proof}

TCSs are CSs, thus the definition of 01-CS carries over.
In particular, a step in the 01-CS exists if there exists a corresponding step in the TCS.
However, the 01-CS of a TCS can also be characterized in the following way. 

\begin{restatable}{lemma}{LmAbstractTCS}
\label{lm:abstract-system}
Let $\csystem$ be a TCS given by $(Q, \delta, \Dmin)$, 
and $\csystemabstract$ its 01-CS.
There is a step $\conf{v}^\alpha \trans{} \conf{w}^\alpha$ in $\csystemabstract$
if and only if there exists $D \in \Dmin$ such that 
$\supp{\pre D} \subseteq \supp{\conf{v}^\alpha}$ and 
$\conf{w}^\alpha$ is such that 
(a) $\conf{w}^\alpha(q)$ equals 0 or 1 if $q \in \supp{\pre D} \setminus \supp{\post D}$, 
(b) $\conf{w}^\alpha(q)$ equals 1 if $q \in \supp{\post D}$, and
(c) $\conf{w}^\alpha(q)$ equals $\conf{v}^\alpha(q)$ otherwise.
\end{restatable}

\begin{proof}
Let $\csystem= (Q, \delta, \Dmin)$ be a transition counter system,
$\csystemabstract$ its 01-CS and $\conf{v}^\alpha$ a configuration of $\csystemabstract$.
Suppose there exists $D \in \Dmin$ such that
$\supp{\pre D} \subseteq \supp{\conf{v}^\alpha}$.
Applying $D$ to any $\conf{v}$ of $\csystem$ such that $\alpha(\conf{v}) = \conf{v}^\alpha$
and with enough processes so that $D$ can be applied
always yields a $\conf{w}$ whose image by $\alpha$ verifies points $b)$ and $c)$.
The subtlety lies in point $a)$.

Let $\conf{v}_1$ be the minimal configuration of $\csystem$ such that
$\alpha(\conf{v}_1) = \conf{v}^\alpha$ and such that
$D$ can be applied to $\conf{v}_1$,
i.e., $\conf{v}_1(p)=\pre D (p)$ for all $p \in Q$.
Let $\conf{w}_1$ be the configuration obtained by applying $D$ to $\conf{v}_1$.
Then $\alpha(\conf{w}_1)=\conf{w}_1^\alpha$ is such that $a),b),c)$ are verified,
with $\conf{w}_1^\alpha(q)=0$ for all $q \in\supp{\pre D} \setminus \supp{\post D}$.
Step $\conf{v}_1 \trans{} \conf{w}_1$ implies
step $\conf{v}^\alpha \trans{} \conf{w}_1^\alpha$ in $\csystemabstract$.

Let $\conf{v}_2$ be the  configuration of $\csystem$ equal to
$\conf{v}_1+\conf{q_2}$ for a $q_2 \in \supp{\pre D} \setminus \supp{\post D}$.
It is still that case that $\alpha(\conf{v}_2) = \conf{v}^\alpha$ and that
$D$ can be applied to $\conf{v}_2$.
Let $\conf{w}_2$ be the configuration obtained by applying $D$ to $\conf{v}_2$.
Then $\alpha(\conf{w}_2)=\conf{w}_2^\alpha$ is such that $a),b),c)$ are verified,
with $\conf{w}_2^\alpha(q)=0$ for all $q \in\supp{\pre D} \setminus \supp{\post D}$
except for $q_2$, for which $\conf{w}_2^\alpha(q_2)=1$.
Step $\conf{v}_2 \trans{} \conf{w}_2$ implies
step $\conf{v}^\alpha \trans{} \conf{w}_2^\alpha$ in $\csystemabstract$.
We can repeat this proof idea for any configuration
$\conf{v}_1+\sum_{q \in S} \conf{q}$,
for any subset $S$ of $\supp{\pre D} \setminus \supp{\post D}$,
to obtain all the 0, 1 combinations for $\conf{w}$ to verify $a)$.

Now for the other direction,
there exists a step $\conf{v}^\alpha \trans{} \conf{w}^\alpha$ in $\csystemabstract$
if there exists a step $\conf{v} \trans{} \conf{w}$ in $\csystem$
for $\conf{v}$ such that $\alpha(\conf{v}) = \conf{v}^\alpha$ and
$\conf{w}$ such that $\alpha(\conf{w}) = \conf{w}^\alpha$.
By definition of a transition counter system,
there exists $D'\in \N^\delta$ and $D \in \Dmin$, such that
$\supp{D'}=\supp{D}$ and $\conf{w}$ is obtained by applying $D'$ to $\conf{v}$.
Multiset $D'$ is such that  $\supp{\pre{D'}} \subseteq \supp{\conf{v}^\alpha}$
since $D'$ can be applied to $\conf{v}$ and
$\alpha(\conf{v}) = \conf{v}^\alpha$ implies
$\supp{\conf{v}}=\supp{\conf{v}^\alpha}$.
Since $\supp{\pre{D'}}=\supp{\pre{D}}$,
we have $\supp{\pre{D}} \subseteq \supp{\conf{v}^\alpha}$.
It is clear that $\conf{w}$ obtained by applying $D'$ is such that
$\alpha(\conf{w})$ verifies the conditions $a),b),c)$.
\qed
\end{proof}

This lemma entails that one can check the existence of a step 
in the 01-CS of a TCS 
in polynomial time in the size of $Q$ and $\Dmin$.
This allows us to extend 
\cite{DelzannoSTZ12}'s results.

\begin{restatable}{theorem}{ThmReachGeq}
\label{thm:reach-geqa}
    Given a TCS, deciding CRP for $CC[\geq a]$ is \Polytime-complete.
\end{restatable}

\begin{proof}[Sketch]
Let $\csystem$ be a TCS, $\csystemabstract$ its 01-CS and $\varphi \in CC[\geq a]$.
By \cref{lm:reduce}, the problem can be reduced to checking if 
there is a reachable configuration $\absconf{v}$ in $\csystemabstract$ that satisfies $\varphi_\alpha$. 
Consider the following algorithm: 
start a run in the initial configuration $\conf{v}_0^\alpha$ containing the maximum number of ones, 
i.e. $\conf{v}_0^\alpha(q)=1$ iff $q \in Q_0$.
By \cref{lm:abstract-system} it is possible to only take steps 
that do not decrease the set of states with ones.
This defines a maximal run $\conf{v}_0^\alpha \trans{} \ldots \trans{} \conf{v}_n^\alpha$
of length at most $|Q|$
such that $\conf{v}_n^\alpha(q)=1$ for all $q$ reachable in $\csystemabstract$.
It then suffices to check whether $\conf{v}_n^\alpha \models \varphi_\alpha$.
\Polytime-hardness follows from \Polytime-hardness of CRP for $CC[\geq 1]$ in RBN~\cite{DelzannoSTZ12}, which is a special case of this problem. The full proof can be found in \cref{app:small-tcs}.
\qed
\end{proof}

In the following, we write $\conf{v}^\alpha \trans{D} \conf{w}^\alpha$ for a step as defined in the end of the last proof. 

\begin{restatable}{theorem}{ThmReachGeqZero}
\label{thm:reach-geqa0}
    Given a  TCS, deciding CRP for $CC[\geq a, =0]$ is \NP-complete.
\end{restatable}

\begin{proof}[Sketch]
Let $\csystem$ be a TCS, $\csystemabstract$ its 01-CS and $\varphi \in CC[\geq a, =0]$.
By \cref{lm:reduce}, it suffices to check if there is a $\absconf{v} \models \varphi_\alpha$  
initially reachable in $\csystemabstract$. 
Consider the following (informal) non-deterministic algorithm:
we guess a run $\conf{v}_0^\alpha \trans{} \ldots \trans{} \conf{v}_m^\alpha$ in two parts, 
first guessing a prefix  that increases the set of states with ones, 
 then guessing a suffix that decreases the set of states with ones.
 It then suffices to check whether $\conf{v}_m^\alpha \models \varphi_\alpha$, and the run is of length at most $2|Q|$. 
\NP-hardness follows from \NP-hardness of CRP for $CC[\geq 1,=0]$ in RBN~\cite{DelzannoSTZ12}, which is a special case of this problem.
The full proof can be found in \cref{app:small-tcs}.
\qed
\end{proof}

\begin{remark}
Given a CS, the \emph{deadlock} problem asks 
whether there is a reachable configuration from which no further step can be taken.
In a TCS $\csystem$ where for all minimal steps $D \in \Dmin$ 
there is at most one transition starting in each state, i.e.,
$\pre{D}(p)\le 1, \forall p\in Q$, 
the deadlock problem is solvable in the abstract system $\csystemabstract$.
Indeed, it can be expressed as a CRP problem with cardinality constraint
$\bigwedge_{D \in \Dmin} \bigvee_{q \in \pre{D}} \#q =0$. 
\end{remark}
\cref{tab:CRP} summarizes our results on the CRP and compares them to existing results.
\begin{table*}
\caption{Decidability and Complexity of the Constraint Reachability Problem (CRP)} 
\label{tab:CRP}
\resizebox{\textwidth}{!}{
\begin{tabular}{l|l|l||l|l|l}
\toprule
\multicolumn{3}{c}{Our Results} & \multicolumn{3}{c}{Existing Results}\\
System Class & Constraint Class & Result & System Class & Constraint Class & Results\\
\midrule
\textbf{$\mathbf{\wqo}$-compatible} & $\mathbf{CC[\control,\geq a, =0]}$ &  \PSPACE-complete &  ASM & $CC[\control]$ & co-NP-complete~\cite{EsparzaGM16}\\
\textbf{systems} & & (\cref{thm:CRP-PSPACE-complete})&  disjunctive & $CC[\control,\geq a]$ & decidable~\cite{EmersonK03}\\
& & & & $CC[\control,\geq a,=0]$ & in EXPTIME~\cite{JacobsSV22}\\
\midrule
\textbf{TCS} & $CC[\geq a]$ & \Polytime-complete & RBN & $CC[\geq 1]$ & \Polytime-complete~\cite{DelzannoSTZ12}\\
& & (\cref{thm:reach-geqa}) & disjunctive & $CC[\geq 1]$ & in \Polytime~\cite{Gandalf21}\\

\midrule
\textbf{TCS} & $CC[\geq a,=0]$ & \NP-complete & RBN & $CC[\geq 1,=0]$ & \NP-complete~\cite{DelzannoSTZ12}\\
& & (\cref{thm:reach-geqa0}) & disjunctive & $CC[\geq 1,=0]$ & in \NP~\cite{Gandalf21}\\

\bottomrule
\end{tabular}
}\vskip-0.2cm
\end{table*}

\section{Conclusion}

In this paper, we characterized parameterized systems for which $(0,1)$-counter abstraction is precise, i.e., a safety property holds in the parameterized system if and only if it holds in its $01$-counter system.
Several system models from the literature fall into this class, including reconfigurable broadcast networks, \disjunctive{}s, and asynchronous shared memory protocols.
Our common framework for these systems provides a simpler explanation for existing decidability results, and also extends and improves them.
We prove that the constraint reachability problem for the whole class of systems is \PSPACE-complete (even without a controller process), and that lower complexity bounds can be obtained under additional assumptions.

Note that weaker versions of \cref{lem:reachability,lm:reduce} directly follow from the fact that $(\csystem,\wqo)$ is a well-structured transition system (cf.~\cite{AbdullaCJT00,FinkelS01}):
in these systems, infinite upward-closed sets (as defined by a constraint in $CC[\geq a,=0]$) can be represented by a finite basis wrt. $\wqo$, resulting in a parameterized model checking algorithm with guaranteed termination. 
However, the complexity bound of the general algorithm is huge (e.g., for broadcast protocols~\cite{EsparzaFM99} it has Ackermannian complexity~\cite{SchmitzS13}).
Instead of relying only on this, we introduce a novel argument that directly connects $\wqo$-compatible systems to the $01$-counter system.

In addition to the questions it answers, we think that this work also raises lots of interesting questions, and that it can serve as the basis of a more systematic study of the systems covered in our framework: 
while in this paper we have focused on reachability and safety properties, we conjecture that our framework can also be extended to liveness and termination properties, possibly under additional restrictions on the systems.
Moreover, an extension to more powerful system models may be possible, for example to processes that are timed automata (like in~\cite{AndreEJK24}) or pushdown automata (like in~\cite{EsparzaGM16}).

\begin{credits}
	\subsubsection{\ackname}
	We thank Javier Esparza and Pierre Ganty for many helpful discussions at the start of this paper.
	P. Eichler carried out this work as a member of the Saarbr\"ucken Graduate School of Computer Science.
	This research was funded in part by the German Research Foundation (DFG) grant GSP\&Co (No. 497132954).
	C. Weil-Kennedy's work was supported by the grant PID2022-138072OB-I00, funded by
	MCIN, FEDER, UE and partially supported by PRODIGY Project (TED2021-132464B-I00) funded
	by MCIN and the European Union NextGeneration.
\end{credits}

\newpage
\bibliographystyle{splncs04}
\bibliography{main}

@inproceedings{abdulla1996general,
  author    = {Parosh Aziz Abdulla and
               Karlis Cerans and
               Bengt Jonsson and
               Yih{-}Kuen Tsay},
  bibsource = {dblp computer science bibliography, https://dblp.org},
  booktitle = {Proceedings, 11th Annual {IEEE} Symposium on Logic in Computer Science,
               New Brunswick, New Jersey, USA, July 27-30, 1996},
  doi       = {10.1109/LICS.1996.561359},
  pages     = {313--321},
  publisher = {{IEEE} Computer Society},
  title     = {General Decidability Theorems for Infinite-State Systems},
  year      = {1996}
}

@article{AbdullaCJT00,
  author    = {Parosh Aziz Abdulla and
               Karlis Cerans and
               Bengt Jonsson and
               Yih{-}Kuen Tsay},
  bibsource = {dblp computer science bibliography, https://dblp.org},
  doi       = {10.1006/INCO.1999.2843},
  journal   = {Inf. Comput.},
  number    = {1-2},
  pages     = {109--127},
  title     = {Algorithmic Analysis of Programs with Well Quasi-ordered Domains},
  volume    = {160},
  year      = {2000}
}

@article{AminofKRSV18,
  author    = {Benjamin Aminof and
               Tomer Kotek and
               Sasha Rubin and
               Francesco Spegni and
               Helmut Veith},
  bibsource = {dblp computer science bibliography, https://dblp.org},
  doi       = {10.1007/S00446-017-0302-6},
  journal   = {Distributed Comput.},
  number    = {3},
  pages     = {187--222},
  title     = {Parameterized model checking of rendezvous systems},
  volume    = {31},
  year      = {2018}
}

@inproceedings{AndreEJK24,
  author    = {{\'{E}}tienne Andr{\'{e}} and
               Paul Eichler and
               Swen Jacobs and
               Shyam Lal Karra},
  bibsource = {dblp computer science bibliography, https://dblp.org},
  booktitle = {Verification, Model Checking, and Abstract Interpretation - 25th International
               Conference, {VMCAI} 2024, London, United Kingdom, January 15-16, 2024,
               Proceedings, Part {I}},
  doi       = {10.1007/978-3-031-50524-9\_6},
  editor    = {Rayna Dimitrova and
               Ori Lahav and
               Sebastian Wolff},
  pages     = {124--146},
  publisher = {Springer},
  series    = {Lecture Notes in Computer Science},
  title     = {Parameterized Verification of Disjunctive Timed Networks},
  volume    = {14499},
  year      = {2024}
}

@article{AngluinAER07,
  author    = {Dana Angluin and
               James Aspnes and
               David Eisenstat and
               Eric Ruppert},
  bibsource = {dblp computer science bibliography, https://dblp.org},
  doi       = {10.1007/S00446-007-0040-2},
  journal   = {Distributed Comput.},
  number    = {4},
  pages     = {279--304},
  title     = {The computational power of population protocols},
  volume    = {20},
  year      = {2007}
}

@article{Apt86,
  author    = {Krzysztof R. Apt and
               Dexter Kozen},
  bibsource = {dblp computer science bibliography, https://dblp.org},
  doi       = {10.1016/0020-0190(86)90071-2},
  journal   = {Inf. Process. Lett.},
  number    = {6},
  pages     = {307--309},
  title     = {Limits for Automatic Verification of Finite-State Concurrent Systems},
  volume    = {22},
  year      = {1986}
}

@inproceedings{AusserlechnerJK16,
  author    = {Simon Au{\ss}erlechner and
               Swen Jacobs and
               Ayrat Khalimov},
  bibsource = {dblp computer science bibliography, https://dblp.org},
  booktitle = {Verification, Model Checking, and Abstract Interpretation - 17th International
               Conference, {VMCAI} 2016, St. Petersburg, FL, USA, January 17-19,
               2016. Proceedings},
  doi       = {10.1007/978-3-662-49122-5\_23},
  editor    = {Barbara Jobstmann and
               K. Rustan M. Leino},
  pages     = {476--494},
  publisher = {Springer},
  series    = {Lecture Notes in Computer Science},
  title     = {Tight Cutoffs for Guarded Protocols with Fairness},
  volume    = {9583},
  year      = {2016}
}

@inproceedings{Balasubramanian18,
  author    = {A. R. Balasubramanian and
               Nathalie Bertrand and
               Nicolas Markey},
  bibsource = {dblp computer science bibliography, https://dblp.org},
  booktitle = {Tools and Algorithms for the Construction and Analysis of Systems
               - 24th International Conference, {TACAS} 2018, Held as Part of the
               European Joint Conferences on Theory and Practice of Software, {ETAPS}
               2018, Thessaloniki, Greece, April 14-20, 2018, Proceedings, Part {II}},
  doi       = {10.1007/978-3-319-89963-3\_3},
  editor    = {Dirk Beyer and
               Marieke Huisman},
  pages     = {38--54},
  publisher = {Springer},
  series    = {Lecture Notes in Computer Science},
  title     = {Parameterized Verification of Synchronization in Constrained Reconfigurable
               Broadcast Networks},
  volume    = {10806},
  year      = {2018}
}

@inproceedings{Balasubramanian22,
  author    = {A. R. Balasubramanian and
               Lucie Guillou and
               Chana Weil{-}Kennedy},
  bibsource = {dblp computer science bibliography, https://dblp.org},
  booktitle = {Foundations of Software Science and Computation Structures - 25th
               International Conference, {FOSSACS} 2022, Held as Part of the European
               Joint Conferences on Theory and Practice of Software, {ETAPS} 2022,
               Munich, Germany, April 2-7, 2022, Proceedings},
  doi       = {10.1007/978-3-030-99253-8\_4},
  editor    = {Patricia Bouyer and
               Lutz Schr{\"{o}}der},
  pages     = {61--80},
  publisher = {Springer},
  series    = {Lecture Notes in Computer Science},
  title     = {Parameterized Analysis of Reconfigurable Broadcast Networks},
  volume    = {13242},
  year      = {2022}
}

@inproceedings{BaumeisterEJSV24,
  author    = {Tom Baumeister and
               Paul Eichler and
               Swen Jacobs and
               Mouhammad Sakr and
               Marcus V{\"{o}}lp},
  bibsource = {dblp computer science bibliography, https://dblp.org},
  booktitle = {Formal Methods - 26th International Symposium, {FM} 2024, Milan, Italy,
               September 9-13, 2024, Proceedings, Part {I}},
  doi       = {10.1007/978-3-031-71162-6\_33},
  editor    = {Andr{\'{e}} Platzer and
               Kristin Yvonne Rozier and
               Matteo Pradella and
               Matteo Rossi},
  pages     = {638--657},
  publisher = {Springer},
  series    = {Lecture Notes in Computer Science},
  title     = {Parameterized Verification of Round-Based Distributed Algorithms via
               Extended Threshold Automata},
  volume    = {14933},
  year      = {2024}
}

@article{BertrandDGGG19,
  author    = {Nathalie Bertrand and
               Miheer Dewaskar and
               Blaise Genest and
               Hugo Gimbert and
               Adwait Amit Godbole},
  bibsource = {dblp computer science bibliography, https://dblp.org},
  doi       = {10.23638/LMCS-15(3:6)2019},
  journal   = {Log. Methods Comput. Sci.},
  number    = {3},
  title     = {Controlling a population},
  volume    = {15},
  year      = {2019}
}

@inproceedings{BertrandFS14,
  author    = {Nathalie Bertrand and
               Paulin Fournier and
               Arnaud Sangnier},
  bibsource = {dblp computer science bibliography, https://dblp.org},
  booktitle = {Foundations of Software Science and Computation Structures - 17th
               International Conference, {FOSSACS} 2014, Held as Part of the European
               Joint Conferences on Theory and Practice of Software, {ETAPS} 2014,
               Grenoble, France, April 5-13, 2014, Proceedings},
  doi       = {10.1007/978-3-642-54830-7\_9},
  editor    = {Anca Muscholl},
  pages     = {134--148},
  publisher = {Springer},
  series    = {Lecture Notes in Computer Science},
  title     = {Playing with Probabilities in Reconfigurable Broadcast Networks},
  volume    = {8412},
  year      = {2014}
}

@inproceedings{BertrandFS15,
  author    = {Nathalie Bertrand and
               Paulin Fournier and
               Arnaud Sangnier},
  bibsource = {dblp computer science bibliography, https://dblp.org},
  booktitle = {26th International Conference on Concurrency Theory, {CONCUR} 2015,
               Madrid, Spain, September 1.4, 2015},
  doi       = {10.4230/LIPICS.CONCUR.2015.44},
  editor    = {Luca Aceto and
               David de Frutos{-}Escrig},
  pages     = {44--57},
  publisher = {Schloss Dagstuhl - Leibniz-Zentrum f{\"{u}}r Informatik},
  series    = {LIPIcs},
  title     = {Distributed Local Strategies in Broadcast Networks},
  volume    = {42},
  year      = {2015}
}

@inproceedings{BouyerMRSS16,
  author    = {Patricia Bouyer and
               Nicolas Markey and
               Mickael Randour and
               Arnaud Sangnier and
               Daniel Stan},
  bibsource = {dblp computer science bibliography, https://dblp.org},
  booktitle = {43rd International Colloquium on Automata, Languages, and Programming,
               {ICALP} 2016, July 11-15, 2016, Rome, Italy},
  doi       = {10.4230/LIPIcs.ICALP.2016.106},
  editor    = {Ioannis Chatzigiannakis and
               Michael Mitzenmacher and
               Yuval Rabani and
               Davide Sangiorgi},
  pages     = {106:1--106:14},
  publisher = {Schloss Dagstuhl - Leibniz-Zentrum f{\"{u}}r Informatik},
  series    = {LIPIcs},
  title     = {Reachability in Networks of Register Protocols under Stochastic Schedulers},
  volume    = {55},
  year      = {2016}
}

@article{ColcombetFO21,
  author    = {Thomas Colcombet and
               Nathana{\"{e}}l Fijalkow and
               Pierre Ohlmann},
  doi       = {10.46298/LMCS-17(4:12)2021},
  journal   = {Log. Methods Comput. Sci.},
  number    = {4},
  title     = {Controlling a random population},
  volume    = {17},
  year      = {2021}
}

@inproceedings{conf/focs/Kozen77,
  author    = {Dexter Kozen},
  bibsource = {dblp computer science bibliography, https://dblp.org},
  booktitle = {18th Annual Symposium on Foundations of Computer Science, Providence,
               Rhode Island, USA, 31 October - 1 November 1977},
  doi       = {10.1109/SFCS.1977.16},
  pages     = {254--266},
  publisher = {{IEEE} Computer Society},
  title     = {Lower Bounds for Natural Proof Systems},
  year      = {1977}
}

@inproceedings{DBLP:conf/popl/EmersonN95,
  author    = {E. Allen Emerson and
               Kedar S. Namjoshi},
  bibsource = {dblp computer science bibliography, https://dblp.org},
  booktitle = {Conference Record of POPL'95: 22nd {ACM} {SIGPLAN-SIGACT} Symposium
               on Principles of Programming Languages, San Francisco, California,
               USA, January 23-25, 1995},
  doi       = {10.1145/199448.199468},
  editor    = {Ron K. Cytron and
               Peter Lee},
  pages     = {85--94},
  publisher = {{ACM} Press},
  title     = {Reasoning about Rings},
  year      = {1995}
}

@article{DBLP:journals/iandc/Finkel90,
  author    = {Alain Finkel},
  bibsource = {dblp computer science bibliography, https://dblp.org},
  journal   = {Inf. Comput.},
  number    = {2},
  pages     = {144--179},
  title     = {Reduction and covering of infinite reachability trees},
  url       = {https://doi.org/10.1016/0890-5401(90)90009-7},
  volume    = {89},
  year      = {1990}
}

@book{DBLP:series/synthesis/2015Bloem,
  author    = {Roderick Bloem and
               Swen Jacobs and
               Ayrat Khalimov and
               Igor Konnov and
               Sasha Rubin and
               Helmut Veith and
               Josef Widder},
  bibsource = {dblp computer science bibliography, https://dblp.org},
  doi       = {10.2200/S00658ED1V01Y201508DCT013},
  publisher = {Morgan {\&} Claypool Publishers},
  series    = {Synthesis Lectures on Distributed Computing Theory},
  title     = {Decidability of Parameterized Verification},
  year      = {2015}
}

@inproceedings{DelzannoSTZ12,
  author    = {Giorgio Delzanno and
               Arnaud Sangnier and
               Riccardo Traverso and
               Gianluigi Zavattaro},
  bibsource = {dblp computer science bibliography, https://dblp.org},
  booktitle = {{IARCS} Annual Conference on Foundations of Software Technology and
               Theoretical Computer Science, {FSTTCS} 2012, December 15-17, 2012,
               Hyderabad, India},
  doi       = {10.4230/LIPICS.FSTTCS.2012.289},
  editor    = {Deepak D'Souza and
               Telikepalli Kavitha and
               Jaikumar Radhakrishnan},
  pages     = {289--300},
  publisher = {Schloss Dagstuhl - Leibniz-Zentrum f{\"{u}}r Informatik},
  series    = {LIPIcs},
  title     = {On the Complexity of Parameterized Reachability in Reconfigurable
               Broadcast Networks},
  volume    = {18},
  year      = {2012}
}

@inproceedings{DelzannoSZ10,
  author    = {Giorgio Delzanno and
               Arnaud Sangnier and
               Gianluigi Zavattaro},
  bibsource = {dblp computer science bibliography, https://dblp.org},
  booktitle = {{CONCUR} 2010 - Concurrency Theory, 21th International Conference,
               {CONCUR} 2010, Paris, France, August 31-September 3, 2010. Proceedings},
  doi       = {10.1007/978-3-642-15375-4\_22},
  editor    = {Paul Gastin and
               Fran{\c{c}}ois Laroussinie},
  pages     = {313--327},
  publisher = {Springer},
  series    = {Lecture Notes in Computer Science},
  title     = {Parameterized Verification of Ad Hoc Networks},
  volume    = {6269},
  year      = {2010}
}

@inproceedings{DelzannoSZ12,
  author    = {Giorgio Delzanno and
               Arnaud Sangnier and
               Gianluigi Zavattaro},
  bibsource = {dblp computer science bibliography, https://dblp.org},
  booktitle = {Formal Techniques for Distributed Systems - Joint 14th {IFIP} {WG}
               6.1 International Conference, {FMOODS} 2012 and 32nd {IFIP} {WG} 6.1
               International Conference, {FORTE} 2012, Stockholm, Sweden, June 13-16,
               2012. Proceedings},
  doi       = {10.1007/978-3-642-30793-5\_15},
  editor    = {Holger Giese and
               Grigore Rosu},
  pages     = {235--250},
  publisher = {Springer},
  series    = {Lecture Notes in Computer Science},
  title     = {Verification of Ad Hoc Networks with Node and Communication Failures},
  volume    = {7273},
  year      = {2012}
}

@inproceedings{EmersonK00,
  author    = {E. Allen Emerson and
               Vineet Kahlon},
  bibsource = {dblp computer science bibliography, https://dblp.org},
  booktitle = {Automated Deduction - CADE-17, 17th International Conference on Automated
               Deduction, Pittsburgh, PA, USA, June 17-20, 2000, Proceedings},
  doi       = {10.1007/10721959\_19},
  editor    = {David A. McAllester},
  pages     = {236--254},
  publisher = {Springer},
  series    = {Lecture Notes in Computer Science},
  title     = {Reducing Model Checking of the Many to the Few},
  volume    = {1831},
  year      = {2000}
}

@inproceedings{EmersonK03,
  author    = {E. Allen Emerson and
               Vineet Kahlon},
  bibsource = {dblp computer science bibliography, https://dblp.org},
  booktitle = {18th {IEEE} Symposium on Logic in Computer Science {(LICS} 2003),
               22-25 June 2003, Ottawa, Canada, Proceedings},
  doi       = {10.1109/LICS.2003.1210076},
  pages     = {361--370},
  publisher = {{IEEE} Computer Society},
  title     = {Model Checking Guarded Protocols},
  year      = {2003}
}

@inproceedings{EmersonN98,
  author    = {E. Allen Emerson and
               Kedar S. Namjoshi},
  bibsource = {dblp computer science bibliography, https://dblp.org},
  booktitle = {Thirteenth Annual {IEEE} Symposium on Logic in Computer Science, Indianapolis,
               Indiana, USA, June 21-24, 1998},
  doi       = {10.1109/LICS.1998.705644},
  pages     = {70--80},
  publisher = {{IEEE} Computer Society},
  title     = {On Model Checking for Non-Deterministic Infinite-State Systems},
  year      = {1998}
}

@inproceedings{Esparza14,
  author    = {Javier Esparza},
  bibsource = {dblp computer science bibliography, https://dblp.org},
  booktitle = {31st International Symposium on Theoretical Aspects of Computer Science
               {(STACS} 2014), {STACS} 2014, March 5-8, 2014, Lyon, France},
  doi       = {10.4230/LIPICS.STACS.2014.1},
  editor    = {Ernst W. Mayr and
               Natacha Portier},
  pages     = {1--10},
  publisher = {Schloss Dagstuhl - Leibniz-Zentrum f{\"{u}}r Informatik},
  series    = {LIPIcs},
  title     = {Keeping a Crowd Safe: On the Complexity of Parameterized Verification
               (Invited Talk)},
  volume    = {25},
  year      = {2014}
}

@inproceedings{EsparzaFM99,
  author    = {Javier Esparza and
               Alain Finkel and
               Richard Mayr},
  bibsource = {dblp computer science bibliography, https://dblp.org},
  booktitle = {14th Annual {IEEE} Symposium on Logic in Computer Science, Trento,
               Italy, July 2-5, 1999},
  doi       = {10.1109/LICS.1999.782630},
  pages     = {352--359},
  publisher = {{IEEE} Computer Society},
  title     = {On the Verification of Broadcast Protocols},
  year      = {1999}
}

@article{EsparzaGM16,
  author    = {Javier Esparza and
               Pierre Ganty and
               Rupak Majumdar},
  bibsource = {dblp computer science bibliography, https://dblp.org},
  doi       = {10.1145/2842603},
  journal   = {J. {ACM}},
  number    = {1},
  pages     = {10:1--10:48},
  title     = {Parameterized Verification of Asynchronous Shared-Memory Systems},
  volume    = {63},
  year      = {2016}
}

@inproceedings{EsparzaRW19,
  author    = {Javier Esparza and
               Mikhail A. Raskin and
               Chana Weil{-}Kennedy},
  bibsource = {dblp computer science bibliography, https://dblp.org},
  booktitle = {Application and Theory of Petri Nets and Concurrency - 40th International
               Conference, {PETRI} {NETS} 2019, Aachen, Germany, June 23-28, 2019,
               Proceedings},
  doi       = {10.1007/978-3-030-21571-2\_20},
  editor    = {Susanna Donatelli and
               Stefan Haar},
  pages     = {365--385},
  publisher = {Springer},
  series    = {Lecture Notes in Computer Science},
  title     = {Parameterized Analysis of Immediate Observation Petri Nets},
  volume    = {11522},
  year      = {2019}
}

@article{FinkelS01,
  author    = {Alain Finkel and
               Philippe Schnoebelen},
  bibsource = {dblp computer science bibliography, https://dblp.org},
  doi       = {10.1016/S0304-3975(00)00102-X},
  journal   = {Theor. Comput. Sci.},
  number    = {1-2},
  pages     = {63--92},
  title     = {Well-structured transition systems everywhere!},
  volume    = {256},
  year      = {2001}
}

@inproceedings{Gandalf21,
  author    = {A. R. Balasubramanian and
               Chana Weil{-}Kennedy},
  bibsource = {dblp computer science bibliography, https://dblp.org},
  booktitle = {Proceedings 12th International Symposium on Games, Automata, Logics,
               and Formal Verification, GandALF 2021, Padua, Italy, 20-22 September
               2021},
  doi       = {10.4204/EPTCS.346.2},
  editor    = {Pierre Ganty and
               Davide Bresolin},
  pages     = {18--34},
  series    = {{EPTCS}},
  title     = {Reconfigurable Broadcast Networks and Asynchronous Shared-Memory Systems
               are Equivalent},
  volume    = {346},
  year      = {2021}
}

@article{GermanS92,
  author    = {Steven M. German and
               A. Prasad Sistla},
  bibsource = {dblp computer science bibliography, https://dblp.org},
  doi       = {10.1145/146637.146681},
  journal   = {J. {ACM}},
  number    = {3},
  pages     = {675--735},
  title     = {Reasoning about Systems with Many Processes},
  volume    = {39},
  year      = {1992}
}

@inproceedings{GuillouSS23,
  author    = {Lucie Guillou and
               Arnaud Sangnier and
               Nathalie Sznajder},
  bibsource = {dblp computer science bibliography, https://dblp.org},
  booktitle = {34th International Conference on Concurrency Theory, {CONCUR} 2023,
               September 18-23, 2023, Antwerp, Belgium},
  doi       = {10.4230/LIPIcs.CONCUR.2023.7},
  editor    = {Guillermo A. P{\'{e}}rez and
               Jean{-}Fran{\c{c}}ois Raskin},
  pages     = {7:1--7:17},
  publisher = {Schloss Dagstuhl - Leibniz-Zentrum f{\"{u}}r Informatik},
  series    = {LIPIcs},
  title     = {Safety Analysis of Parameterised Networks with Non-Blocking Rendez-Vous},
  volume    = {279},
  year      = {2023}
}

@inproceedings{JacobsS18,
  author    = {Swen Jacobs and
               Mouhammad Sakr},
  bibsource = {dblp computer science bibliography, https://dblp.org},
  booktitle = {Verification, Model Checking, and Abstract Interpretation - 19th International
               Conference, {VMCAI} 2018, Los Angeles, CA, USA, January 7-9, 2018,
               Proceedings},
  doi       = {10.1007/978-3-319-73721-8\_12},
  editor    = {Isil Dillig and
               Jens Palsberg},
  pages     = {247--268},
  publisher = {Springer},
  series    = {Lecture Notes in Computer Science},
  title     = {Analyzing Guarded Protocols: Better Cutoffs, More Systems, More Expressivity},
  volume    = {10747},
  year      = {2018}
}

@inproceedings{JacobsSV22,
  author    = {Swen Jacobs and
               Mouhammad Sakr and
               Marcus V{\"{o}}lp},
  bibsource = {dblp computer science bibliography, https://dblp.org},
  booktitle = {22nd Formal Methods in Computer-Aided Design, {FMCAD} 2022, Trento,
               Italy, October 17-21, 2022},
  doi       = {10.34727/2022/ISBN.978-3-85448-053-2\_29},
  editor    = {Alberto Griggio and
               Neha Rungta},
  pages     = {225--234},
  publisher = {{IEEE}},
  title     = {Automatic Repair and Deadlock Detection for Parameterized Systems},
  year      = {2022}
}

@article{Kruskal72,
  author    = {Joseph B. Kruskal},
  bibsource = {dblp computer science bibliography, https://dblp.org},
  doi       = {10.1016/0097-3165(72)90063-5},
  journal   = {J. Comb. Theory, Ser. {A}},
  number    = {3},
  pages     = {297--305},
  title     = {The Theory of Well-Quasi-Ordering: {A} Frequently Discovered Concept},
  volume    = {13},
  year      = {1972}
}

@article{LucaV94,
  author    = {Aldo de Luca and
               Stefano Varricchio},
  bibsource = {dblp computer science bibliography, https://dblp.org},
  doi       = {10.1007/BF01213206},
  journal   = {Acta Informatica},
  number    = {6},
  pages     = {539--557},
  title     = {Well Quasi-Orders and Regular Languages},
  volume    = {31},
  year      = {1994}
}

@inproceedings{PnueliXZ02,
  author    = {Amir Pnueli and
               Jessie Xu and
               Lenore D. Zuck},
  bibsource = {dblp computer science bibliography, https://dblp.org},
  booktitle = {Computer Aided Verification, 14th International Conference, {CAV}
               2002,Copenhagen, Denmark, July 27-31, 2002, Proceedings},
  doi       = {10.1007/3-540-45657-0\_9},
  editor    = {Ed Brinksma and
               Kim Guldstrand Larsen},
  pages     = {107--122},
  publisher = {Springer},
  series    = {Lecture Notes in Computer Science},
  title     = {Liveness with (0, 1, infty)-Counter Abstraction},
  volume    = {2404},
  year      = {2002}
}

@inproceedings{SchmitzS13,
  author    = {Sylvain Schmitz and
               Philippe Schnoebelen},
  bibsource = {dblp computer science bibliography, https://dblp.org},
  booktitle = {{CONCUR} 2013 - Concurrency Theory - 24th International Conference,
               {CONCUR} 2013, Buenos Aires, Argentina, August 27-30, 2013. Proceedings},
  doi       = {10.1007/978-3-642-40184-8\_2},
  editor    = {Pedro R. D'Argenio and
               Hern{\'{a}}n C. Melgratti},
  pages     = {5--24},
  publisher = {Springer},
  series    = {Lecture Notes in Computer Science},
  title     = {The Power of Well-Structured Systems},
  volume    = {8052},
  year      = {2013}
}

@article{Suzuki88,
  author    = {Ichiro Suzuki},
  bibsource = {dblp computer science bibliography, https://dblp.org},
  doi       = {10.1016/0020-0190(88)90211-6},
  journal   = {Inf. Process. Lett.},
  number    = {4},
  pages     = {213--214},
  title     = {Proving Properties of a Ring of Finite-State Machines},
  volume    = {28},
  year      = {1988}
}

@inproceedings{Waldburger23,
  author    = {Nicolas Waldburger},
  bibsource = {dblp computer science bibliography, https://dblp.org},
  booktitle = {48th International Symposium on Mathematical Foundations of Computer
               Science, {MFCS} 2023, August 28 to September 1, 2023, Bordeaux, France},
  doi       = {10.4230/LIPIcs.MFCS.2023.88},
  editor    = {J{\'{e}}r{\^{o}}me Leroux and
               Sylvain Lombardy and
               David Peleg},
  pages     = {88:1--88:15},
  publisher = {Schloss Dagstuhl - Leibniz-Zentrum f{\"{u}}r Informatik},
  series    = {LIPIcs},
  title     = {Checking Presence Reachability Properties on Parameterized Shared-Memory
               Systems},
  volume    = {272},
  year      = {2023}
}

\newpage
\appendix

\section{Additional Examples for Section \ref{sec:prelim}}\label{sec:additional-example-steps}

This section provides additional examples for some of the transitions introduced in \cref{sec:transition-types}.

\begin{figure}[!ht]
    \centering
    \begin{tikzpicture}[font=\footnotesize]
        \node[location,initial] (a1) {$c_1$};
        \node[location] (a2) [right =of a1] {$c_2$};

        \path[->]
        (a1) edge[bend left = 40] node[above,locguard] {$q_1$} (a2)
        (a2) edge[bend left = 40] node[below,locguard] {$q_2$} (a1)
        (a2) edge[loop above] node[above] {} (a2)
        ;
    \end{tikzpicture}
    \hspace{3em}
    \begin{tikzpicture}[font=\footnotesize]
        \node[location,initial] (a1) {$q_1$};
        \node[location] (a2) [right =of a1] {$q_2$};
        \node[location] (a3) [right =of a2] {$q_3$};

        \path[->]
        (a1) edge[bend left = 40] node[above,locguard] {$c_2$} (a2)
        (a2) edge[bend left = 40] node[below] {} (a1)
        (a2) edge node[above,locguard] {$q_2$} (a3)
        (a1) edge[loop above] node[above] {} (a1)
        (a2) edge[loop above] node[above] {} (a2)
        ;
    \end{tikzpicture}
    \caption{A \disjunctive{} with two controller states and three user states.}
    \label{fig:disj}
\end{figure}

\begin{example}
    \cref{fig:disj} depicts a disjunctive guard protocol.
    From initial configuration $(c_1,\conf{v})$ with $\conf{v}=(1,0,0)$, there is a step to
    $(c_2,\conf{v})$: the controller takes transition $c_1 \trans{q_1} c_2$.
    Note that from $(c_1,\conf{v})$, the user process in $q_1$ cannot take transition $q_1 \trans{c_2} q_2$,
    as this would require the controller to be in $c_2$.
\end{example}

\begin{figure}[!ht]
    \centering
    \begin{tikzpicture}[font=\footnotesize]
        \node[location,initial] (a1) {$00$};
        \node[location] (a2) [right =of a1] {$01$};
        \node[location] (a3) [below =of a1] {$10$};
        \node[location] (a4) [right =of a3] {$11$};

        \path[<->]
        (a1) edge[bend left = 40] node[above] {} (a2)
        (a1) edge[bend right = 40] node[above] {} (a3)
        (a1) edge node[above] {} (a4)
        (a2) edge[bend left = 40] node[above] {} (a4)
        (a2) edge node[above] {} (a3)
        (a3) edge[bend right = 40] node[above] {} (a4)
        ;
    \end{tikzpicture}
    \begin{tikzpicture}[font=\footnotesize]
        \node[location,initial] (a1) {$q_1$};
        \node (a) [right =of a1] {};
        \node[location] (a2) [right =of a] {$q_2$};
        \node[location] (a3) [right =of a2] {$q_3$};

        \path[->]
        (a1) edge[bend left = 40] node[above] {$w(10)$} (a2)
        (a2) edge[bend left = 40] node[below] {$r(11)$} (a1)
        (a2) edge node[above] {$r(01)$} (a3)
        (a1) edge[loop above] node[above] {$w(01)$} (a1)
        (a2) edge[loop above] node[above] {$w(11)$} (a2)
        ;
    \end{tikzpicture}

    \caption{An ASM system with four controller states (i.e., variable valuations) and three user states.}
    \label{fig:shared}
\end{figure}

\begin{example}
    \cref{fig:shared} depicts an ASM protocol.
    Starting in initial configuration $(00,(2,0,0))$, there is a run
    $$(00,(2,0,0)) \trans{w(10)} (10,(1,1,0)) \trans{w(01)} (01,(1,1,0)) \trans{r(01)} (01,(1,0,1)).$$
    Notice that from the initial configuration $(00,(1,0,0))$, state $q_3$ is not reachable.
\end{example}

\section{Additional Details for Section \ref{sec:basic}}\label{sec:appendix-sec-compat}

This section provides  the omitted $\wqo$ compatibility proofs from \cref{sec:basic}.

\LmCompatibleAll*

We split the proof of \cref{lm:compatible-all} into several lemmas depending on the step type.
Recall that the following result is already proved in the main text.

\begin{lemma}
\label{lm:compatible-rbn}
    CSs induced by lossy broadcasts are fully $\wqo$-compatible.
\end{lemma}

We show it also holds  for disjunctive guard steps, synchronization steps and ASM steps.

\begin{restatable}{lemma}{LmCompatibleDG}
    \label{lm:compatible-dg}
        CSs induced by disjunctive guard steps
        are fully $\wqo$-compatible.
    \end{restatable}
    \begin{proof}
    To prove forward $\wqo$-compatibility,
    assume there is a step $(c,\conf{v}) \rightarrow (c',\conf{v}')$ and $(c,\conf{v}) \wqo (d,\conf{w})$.
    This step is made up of $k\ge 1$ processes taking a  transition $p \trans{\existsGuard} q$, and is only enabled if at least one process is in $\existsGuard$ in $(c,\conf{v})$.
    Since $(c,\conf{v}) \le (d,\conf{w})$, the guard $\existsGuard$ is also satisfied in $(d,\conf{w})$ and the transition can be taken.
    Distinguish two cases:
    \begin{enumerate}
        \item if $\conf{v}'(p)>0$, then after taking the same $k$ transitions from $(d,\conf{w})$
        we arrive in a configuration $(d',\conf{w}')$ with $(c',\conf{v}') \wqo (d',\conf{w}')$.
        \item if $\conf{v}'(p)=0$,
        then we let all processes that are in $p$ in $(d,\conf{w})$ take the transition $p \trans{\existsGuard} q$, and we reach a configuration $(d',\conf{w}')$ with $\conf{w}'(p')=0$, and which also satisfies $(c',\conf{v}') \wqo (d',\conf{w}')$.
    \end{enumerate}
    
    Backward $\wqo$-compatibility can be proven in a similar way.
    \qed
\end{proof}

\begin{restatable}{lemma}{LmCompatibleSync}
    \label{lm:compatible-synchronizations}
        CSs induced by synchronization steps are fully $\wqo$-compatible. 
    \end{restatable}
    
    \begin{proof}
    To prove forward $\wqo$-compatibility,
    assume there is a step $(c,\conf{v}) \rightarrow (c',\conf{v}')$ and $(c,\conf{v}) \wqo (d,\conf{w})$.
    This step is made up of a letter $a$ and all processes taking a transition of the form $p \trans{a} p'$ if they have one.
    Note that a step is not fully defined by the chosen letter $a$:
    states may have more than one outgoing $a$-transition for the processes to choose from.
    We arrive in a configuration $(d',\conf{w}')$ with $(c',\conf{v}') \wqo (d',\conf{w}')$
    by taking an $a$-step from $(d,\conf{w})$ as follows.
    If $k$ processes take some $p \trans{a} p'$ from $(c,\conf{v})$,
    then $k$ processes also take it from $(d,\conf{w})$.
    If $\conf{v}'(q)=0$ and $\conf{v}(q)\ge 0$ for some state $q$,
    then there is at least one transition from $q$ that is taken in $(c,\conf{v}) \rightarrow (c',\conf{v}')$;
    take it an extra $\conf{w}(q)-\conf{v}(q)$ times to empty $q$ from $(d,\conf{w})$.
    
    Backward $\wqo$-compatibility can be proven in a similar way.
    \qed
\end{proof}

\begin{restatable}{lemma}{LmCompatibleSFD}
    \label{lm:compatible-sfd}
        CSs induced by ASM steps are fully $\wqo$-compatible.
    \end{restatable}
    \begin{proof}
    To prove forward $\wqo$-compatibility,
    assume there is a step $(c,\conf{v}) \rightarrow (c',\conf{v}')$ and $(c,\conf{v}) \wqo (d,\conf{w})$.
    This step is made up of one process taking a  transition $p \trans{l(a)} p'$ with $l\in \set{w,r}$.
    Since $(c,\conf{v}) \le (d,\conf{w})$, this process is also in $(d,\conf{w})$.
    Distinguish two cases:
    \begin{enumerate}
        \item if $\conf{v}'(p)>0$, then after taking the same transition once from $(d,\conf{w})$ we arrive in a configuration $(d',\conf{w}')$ with $(c',\conf{v}') \wqo (d',\conf{w}')$.
        \item if $\conf{v}'(p)=0$, then we take the transition repeatedly until we reach a configuration $(d',\conf{w}')$ with $\conf{w}'(p')=0$:
        either it is a write transition and rewriting the same symbol does not change the shared variable (the state of the controller),
        or it is a read transition and the shared variable can be read again and again.
        This configuration  satisfies $(c',\conf{v}') \wqo (d',\conf{w}')$.
    \end{enumerate}
    Backward $\wqo$-compatibility can be proven in a similar way.
    \qed
\end{proof}

\section{Additional Details for \cref*{sec:reach}}\label{app:reach}

\subsection{CRP}

\begin{example} \label{ex:crp-problems}
    We give some examples of parameterized reachability problems expressed in CRP format.
    \begin{itemize}[topsep=2pt,parsep=1pt]
    \item The \emph{cover} problem (as in \cite{DelzannoSZ10,GuillouSS23,Waldburger23}) asks, given a counter system and a state $q_f \in Q$,
    whether a configuration with a least one process in $q_f$ is reachable.
    This can be expressed as a CRP with cardinality constraint $\#q_f \geq 1$. 
    This problem is also sometimes called control state reachability \cite{DelzannoSTZ12}
    (the ``control state'' is $q_f$ in this case, not to be confused with the state of the controller process).
    
    \item A variant of the cover problem can also be stated with respect to a state $c_f$ of the controller process.
    This can be expressed as a CRP with cardinality constraint $\control = c_f$.
    
    \item The \emph{coverability} problem (in the classic Petri nets sense) asks, given a counter system and a configuration  $(c,\conf{v})$,
    whether a configuration $(c,\conf{w})$ with $\conf{w} \ge \conf{v}$ is reachable.
    This corresponds to a CRP with cardinality constraint
    $\bigwedge_{q \in Q} \#q \geq \conf{v}(q)$.
    
    \item The \emph{target} problem  (as in \cite{DelzannoSZ10,BertrandFS15,Waldburger23}) asks, given a counter system with a distinguished state $q_f$,
    whether a configuration with all processes in $q_f$ is reachable.
    This can be expressed as a CRP with cardinality constraint
    $\bigwedge_{q \neq q_f} \#q = 0$.
    
    \end{itemize}
\end{example}

The constant $B_\csystem$ is usually small in counter systems. 

\begin{remark}
We provide more details on $B_\csystem$ for different types of transitions.
\label{rmk:bc}
\begin{itemize}[topsep=2pt,parsep=1pt]
\item
    Let $\csystem$ be a counter system with only lossy broadcast steps.
    Then $B_\csystem$ is bounded by $|Q|$:
    a step depends on one broadcast transition and an arbitrary number of receive transitions.
    In the worst case, a minimal step is such that, for a given state $p$, the broadcast is $p \trans{!a} p'$
    and there are receive transitions $p \trans{?a} q$ for every $q \in Q \setminus \{ p' \}$.
\item Let $\csystem$ be a counter system with only disjunctive guard steps.
    Then $B_\csystem$ is bounded by $2$:
    a step depends on one process that can take the transition, and one process that satisfies the guard. In the worst case both of them are user processes in the same state.
\item Let $\csystem$ be a counter system with only synchronization steps.
    Then $B_\csystem$ is bounded by $|Q|$:
    a step depends on a subset of the transitions labeled by the same letter.
    In the worst case there are $|Q|$ $a$-labeled transitions leaving from the same state.
\item Let $\csystem$ be a counter system with only ASM steps.
    Then $B_\csystem$ is bounded by $1$:
    a step depends on one controller transition and one user transition.
    \item For a CS with several types of these steps, 
    $B_\csystem$ is bounded by the maximum of the constants given above.
\end{itemize}
\end{remark}

\subsection{\PSPACE-Hardness}\label{app:hardness}

We prove \PSPACE-hardness by a reduction of the intersection non-emptiness problem for deterministic finite automata (DFA) to the CRP. 
Let $\autom_i = (\automStates_i, \automAlphabet,$
$ \automTrans_i, \automInitState_i, \automAcceptStates_i)$ for $i \in \{ 1, \dots, n \}$ be a set of $n$ DFA with common input alphabet $\automAlphabet$.
The intersection non-emptiness problem (INT) asks whether there exists a word $\word \in \automAlphabet^*$ that is accepted by all $n$ automata.
INT is known to be PSPACE-complete~\cite{conf/focs/Kozen77}.

We can directly encode INT into a CRP by considering (arbitrarily many copies of) the automata as communicating with synchronization actions with labels $\automAlphabet$. 
That is, we consider the CS $\csystem_{\autom_1, \dots ,\autom_n} = (\Loc, \Transitions)$ (without a controller), with $\Loc = \automStates_1 \cup \dots \cup \automStates_n$, initial states $\{ \automInitState_1, \dots, \automInitState_n \}$ and $\Transitions = \automTrans_1 \cup \dots \cup \automTrans_n$. 
A constraint that expresses that all the automata are in a final state at the same time is

\begin{align*}
    \varphi = \bigwedge_{\autom_i \in \{ \autom_1, \dots, \autom_n \}} \left( \bigvee_{p \in \automAcceptStates_i} p \geq 1 \right)
\end{align*}

\begin{restatable}{lemma}{LmHardness}
  Let $\autom_1, \dots, \autom_n$, $\csystem_{\autom_1, \dots, \autom_n}$ and $\varphi$ defined as above.
 The $\autom_i$ are assumed to be complete, i.e. for each state $q$ and letter $a$, there is a transition from $p$ reading $a$.
  Then the intersection between $\autom_1, \dots, \autom_n$ is non-empty if and only if a configuration $\conf{v}$ with $\conf{v} \models \varphi$ is reachable in $\csystem_{\autom_1, \dots, \autom_n}$.
\end{restatable}

\begin{proof}
  In  $\csystem_{\autom_1, \dots,  \autom_n}$, each process starts in one of the initial states $\automInitState_i$ with $i \in \{1, \dots, n\}$ and can only progress by synchronizing over some symbol in $\automAlphabet$. 
  The semantics of synchronization forces all other processes to take a transition with the same label, which mimics inputting the respective symbol to all of the DFA simultaneously.
  $\varphi$ is satisfied if and only if for every automaton $\autom_1, \dots, \autom_n$ an accepting state is reached by at least one process, i.e., the executed sequence of actions corresponds to a word accepted by all automata. Consequently, such a run exists iff the intersection of the language of the DFA is non-empty.
  \qed 
\end{proof}
As the constructed $\csystem_{\autom_1, \dots, \autom_n}$ and $\varphi$ are polynomial in the inputs $\autom_1, \dots, \autom_n$, we obtain our hardness result.
We note that the construction does not use a controller process.

\begin{lemma}
    The CRP for fully $\wqo$-compatible systems without a controller process and with $\varphi \in CC[\geq 1]$ is \PSPACE-hard.
\end{lemma}

\section{Additional Details for \cref*{sec:trace}}\label{app:sec-trace}

\LmPMCPinfinite

\begin{proof}
    We directly get $\tracesinf{\csystem} \subseteq \tracesinf{\csystemabstract}$, since for infinite runs the $(0,1)$-abstraction may be an over-approximation of the possible behaviors.
    
    For the other direction we can prove a stronger property, for a restricted version of $\csystemabstract$:
    Let $\csystemabstract'$ be this modification, which is such that it never takes transitions $(c,\conf{v}^{\alpha}) \trans{} (c',\conf{v}'^{\alpha})$ such that $\conf{v}'^{\alpha}(q) < \conf{v}^{\alpha}(q)$ for any $q \in Q$. 
    That is, once a user state $q$ has been reached, it will always remain occupied.
    
    To see that this is always possible, note that whenever there is a step $(c,\conf{v}) \rightarrow (c',\conf{v'})$ based on a transition $q \trans{\existsGuard} q'$ with $q \in Q$ in $\csystem$, then there is also a step $(c,\conf{v}+\conf{q}) \rightarrow (c',\conf{v'}+\conf{q})$.
    Moreover, since $\alpha(c,\conf{v}) = \alpha(c,\conf{v}+\conf{q})$, there is a step $\alpha(c,\conf{v}) \rightarrow \alpha(c',\conf{v'}+\conf{q})$ in $\csystemabstract$, i.e., where $q$ remains occupied.
    
    As the runs of $\csystemabstract'$ always keep user states occupied once they have been reached, and guards are disjunctive, the transitions of the controller that are enabled will clearly be a superset of the transitions that are enabled on any run of $\csystemabstract$ with the same sequence of transitions. That is, $\tracesinf{\csystemabstract'} \supseteq \tracesinf{\csystemabstract}$.
    
    Moreover, even on infinite runs of $\csystemabstract'$, the vector $\conf{v}^{\alpha}$ will only change finitely often, until every $q \in Q$ that is visited in the infinite run has been visited for the first time.
    As $\csystem$ is a \disjunctive{}, all transitions of the controller that will ever be enabled are enabled at that point, and will stay enabled forever in a run where the user processes never move again.
    Based on this observation and the proof idea of \cref{lem:reachability}, it is easy to show that we also have $\tracesinf{\csystem} \supseteq \tracesinf{\csystemabstract'}$.
    \qed
\end{proof}

\section{Additional Details for \cref*{sec:small-tcs}}\label{app:small-tcs}

\ThmReachGeq

\begin{proof}
    Let $\csystem$ be a $\wqo$-compatible transition counter system,
    $\csystemabstract$ its 01-CS and $\varphi \in CC[\geq a]$.
    By \cref{lm:reduce}, the problem can be reduced to checking if
    there is a reachable configuration $\absconf{v}$ in $\csystemabstract$ that satisfies $\varphi_\alpha$.

    Consider the following algorithm:
    Start a run in the initial configuration $\conf{v}_0^\alpha$ containing the maximum number of ones,
    i.e., $\conf{v}_0^\alpha(q)=1$ iff $q \in Q_0$.
    By \cref{lm:abstract-system} it is possible to only take steps
    that do not decrease the set of states with ones.
    Suppose we have a run
    $\conf{v}_0^\alpha \trans{} \ldots \trans{} \conf{v}_i^\alpha$.
    If it exists, take a step to a $\conf{v}_{i+1}^\alpha$ such that \\
    1) \ $\conf{v}_{i+1}^\alpha(q)=1$ if $\conf{v}_i^\alpha(q)=1$ for all $q \in Q$, and \\
    2) \  there is at least one $q' \in Q$
    such that $\conf{v}_{i+1}^\alpha(q')=1$ and $\conf{v}_i^\alpha(q')=0$.\\
    Keep taking such steps until no longer possible, defining a run
    $\conf{v}_0^\alpha \trans{} \conf{v}_1^\alpha \trans{} \ldots \trans{} \conf{v}_n^\alpha$ in $\csystemabstract$.
    If $\conf{v}_n^\alpha \models \varphi_\alpha$ then the algorithm answers yes,
    otherwise it answers no.

    This is a polynomial time algorithm: there are at most $|Q|$ steps in the run,
    and choosing the next step is polynomial in $|\Dmin|$ by Lemma \ref{lm:abstract-system}:
    at $\conf{v}_i^\alpha$ go through the $D \in \Dmin$ until
    $\supp{\pre{D}} \subseteq \supp{\conf{v}_i^\alpha}$,
    then take the step to a $\conf{v}_{i+1}^\alpha$
    such that if $\conf{v}_i(q)=1$ then $\conf{v}_{i+1}(q)=1$.
    If the algorithm answers yes, there is a reachable configuration
    $\conf{v}_n^\alpha \models \varphi_\alpha$ in $\csystemabstract$.
    For the other direction,
    we first make the following claim.
    \begin{claim}
        We say a state $q \in Q$ is \emph{reachable} in $\csystemabstract$ if there exists
        a reachable $\conf{w}_i^\alpha$ such that $\conf{w}_i^\alpha(q)=1$.
        The $\conf{v}_n^\alpha$ given by a run of the algorithm is such that
        $\conf{v}_n^\alpha(q)=1$ for all $q$ reachable in $\csystemabstract$.
    \end{claim}
    Now, suppose the algorithm answers no, i.e., $\conf{v}_n^\alpha \nvDash \varphi_\alpha$.
    Assume for the sake of contradiction that there exists a reachable $\conf{w}_m^\alpha$
    such that $\conf{w}_m^\alpha\models \varphi_\alpha$.
    This implies that there is a $q$ such that $\conf{w}_m^\alpha(q)=1$ but $\conf{v}_n^\alpha(q)=0$, contradicting the claim.

    We now prove the claim.
    Let $q$ be reachable in $\csystemabstract$.
    We reason by induction on the length $i$ of the shortest run $\conf{w}_0^\alpha \trans{}  \ldots \trans{} \conf{w}_i^\alpha$
    such that $q$ is reachable.
    If $i=0$, then  $\conf{w}_0^\alpha(q)=1$ implies $q\in Q_0$.
    Therefore $\conf{v}_0^\alpha(q)=1$ by definition,
    and since the run of the algorithm never decreases the set of states with ones, $\conf{v}_n^\alpha(q)=1$.
    Now suppose the claim is true for $i$, i.e., $\conf{v}_n^\alpha(q)=1$ for every state $q$ such that its shortest run has length at most $i$.
    We show that then it also holds for $i+1$.
    Let $\conf{w}_0^\alpha \trans{}  \ldots \trans{} \conf{w}_{i+1}^\alpha$ be the shortest run such that $q$ is reachable.
    By \cref{lm:abstract-system}, this implies that there exists a $D \in \Dmin$ such that
    $q\in \post{D}$ and all $q' \in \pre{D}$ are present at the previous step of the run, i.e.,  $\conf{w}_i^\alpha(q')=1$.
    By induction hypothesis, $\conf{v}_n^\alpha(q')=1$ for all $q' \in \pre{D}$,
    meaning that the step defined by $D$ is possible from $\conf{v}_n^\alpha$.
    By the algorithm definition, no more steps can be taken that add a new 1,
    so $\conf{v}_n^\alpha(q)$ is already equal to 1.
    \qed
\end{proof}

\ThmReachGeqZero

\begin{proof}
    Let $\csystem$ be a $\wqo$-compatible transition counter system, $\csystemabstract$ its 01-CS and $\varphi \in CC[\geq a, =0]$.
    By \cref{lm:reduce}, the problem can be reduced to checking if
    there is a reachable configuration $\absconf{v}$ in $\csystemabstract$ that satisfies $\varphi_\alpha$.

    Consider the following non-deterministic algorithm,
    where we guess a run in two parts.
    Informally, we first guess a prefix that increases the set of states with ones,
    then guess a suffix that decreases the set of states with ones.
    Guess an initial configuration $\conf{v}_0^\alpha$ (not necessarily with $\supp{v_0}=Q_0$).
    Guess a run $\conf{v}_0^\alpha \trans{} \ldots \trans{}\conf{v}_n^\alpha$ as in \cref{thm:reach-geqa},
    where each step $\conf{v}_i^\alpha \trans{}  \conf{v}_{i+1}^\alpha$ is such that \\
    1) \ $\conf{v}_{i+1}^\alpha(q)=1$ if $\conf{v}_i^\alpha(q)=1$ for all $q \in Q$, and \\
    2) \ there is at least one $q' \in Q$
    such that $\conf{v}_{i+1}^\alpha(q')=1$ and $\conf{v}_i^\alpha(q')=0$. \\
    As in \cref{thm:reach-geqa}, there are at most $|Q|$ such steps possible,
    but here we may choose to stop despite such steps remaining.
    Then we guess a run $\conf{v}_n^\alpha \trans{} \ldots \trans{} \conf{v}^\alpha_{n+m}$
    where each step $\conf{v}_i^\alpha \trans{}  \conf{v}_{i+1}^\alpha$ is such that \\
    3) \ $\conf{v}_{i+1}^\alpha(q)=0$ if $\conf{v}_i^\alpha(q)=0$ for all $q \in Q$, and \\
    4) \  there is at least one $q' \in Q$
    such that $\conf{v}_{i+1}^\alpha(q')=0$ and $\conf{v}_i^\alpha(q')=1$.  \\
    There are at also at most $|Q|$ such steps possible.
    If $ \conf{v}^\alpha_{n+m}\models \varphi_\alpha$ then the algorithm answers yes,
    otherwise it answers no.

    This is a polynomial time algorithm: there are at most $2|Q|$ steps in the run,
    and choosing the next step is polynomial in $|\Dmin|$ as argued in \cref{thm:reach-geqa}.
    Checking if $\conf{v}^\alpha_{n+m}$ satisfies $ \varphi_\alpha$ is also polynomial time.
    If the algorithm answers yes,  there is a reachable configuration that satisfies $\varphi_\alpha$.
    Suppose there exists a run $\conf{w}_0^\alpha \trans{}  \ldots \trans{} \conf{w}_l^\alpha \models \varphi_\alpha$.
    We can extract from it a  run satisfying $1),2)$ by modifying the steps to leave ones behind
    (which is possible by \cref{lm:abstract-system}) and removing loops,
    thus obtaining a run $\conf{w}_0'^\alpha \trans{} \conf{w}_1'^\alpha \trans{}  \ldots \trans{} \conf{w}_n'^\alpha$ of at most $|Q|$ steps,
    with $\conf{w}_0'^\alpha=\conf{w}_0^\alpha$ and $\conf{w}_n'^\alpha=\conf{w}_l^\alpha$.
    For all $q$ such that $\conf{w}_i^\alpha(q)=1$ for some $i \in\set{1,\ldots, l}$,
    $\conf{w}_n'^\alpha(q)=1$, and  $\conf{w}_n'^\alpha \ge \conf{w}_i^\alpha$ for all $i \in\set{1,\ldots, l}$.
    We now want to continue the run in a way satisfying $3),4)$ to empty all states $q$ such that $\conf{w}_n'^\alpha(q)=1$ but $\conf{w}_l^\alpha(q)=0$,
    i.e., $q \in \supp{\conf{w}_n'^\alpha}\setminus \supp{\conf{w}_l^\alpha}$.
    Let $i_0$ be the smallest index of the set
    $\set{\max_{\conf{w}_i^\alpha(q)=1} i \ | \ q \in \supp{\conf{w}_n'^\alpha}\setminus \supp{\conf{w}_l^\alpha}}$,
    i.e., the smallest index such that $\conf{w}_i^\alpha$ is the last index in the original run at which a state $q$ is filled, where $q$ is a state to be emptied.
    Let $D_0\in \Dmin$ such that $\conf{w}_{i_0}^\alpha \trans{D_0} \conf{w}_{i_0+1}^\alpha$.
    Let $\conf{w}_{n+1}'^\alpha$ be the configuration such that $\conf{w}_{n}'^\alpha \trans{D_0} \conf{w}_{n+1}'^\alpha$, which is possible since  $\conf{w}_n'^\alpha \ge \conf{w}_{i_0}^\alpha$.
    Now let $i_1$ be the smallest index of the set
    $\set{\max_{\conf{w}_i^\alpha(q)=1} i \ | \ q \in \supp{\conf{w}_{n+1}'^\alpha}\setminus \supp{\conf{w}_l^\alpha}}$.
    Note that $i_1 \ge i_0$ and
    let $D_1\in \Dmin$ such that $\conf{w}_{i_1}^\alpha \trans{D_1} \conf{w}_{i_1+1}^\alpha$.
    Let $\conf{w}_{n+2}'^\alpha$ be the configuration such that $\conf{w}_{n+1}'^\alpha \trans{D_1} \conf{w}_{n+2}'^\alpha$.
    We proceed in this way until reaching $\conf{w}_{n+m}'^\alpha$ equal to $\conf{w}_l^\alpha$.
    Run $\conf{w}_0'^\alpha \trans{} \dots \trans{} \conf{w}_{n+m}'^\alpha$
    is a valid run of the algorithm, and thus there exists an execution of it that answers yes.
    \qed
\end{proof}

\end{document}